\documentclass[12pt]{elsarticle}

\usepackage{graphicx,hyperref}
\usepackage{listings,subfig}
\usepackage{algorithm,algorithmic}
\usepackage{amsmath,amssymb,amsfonts,amsthm}
\usepackage{tikz,color}
\usepackage{pxfonts}
\usepackage{wrapfig}
\usepackage[margin=2cm]{geometry}

\theoremstyle{definition}
\newtheorem{thm}{Theorem}
\newtheorem{lem}[thm]{Lemma}

\newtheorem{dfn}[thm]{Definition}

\theoremstyle{remark}
\newtheorem{ex}[thm]{Example}
\newtheorem{rem}[thm]{Remark}

\floatname{algorithm}{Procedure}

\newcommand{\Z}{\mathbb{Z}}

\makeatletter
\def\ps@pprintTitle{%
 \let\@oddhead\@empty
 \let\@evenhead\@empty
 \def\@oddfoot{\centerline{\thepage}}%
 \let\@evenfoot\@oddfoot}
\makeatother

\begin{document}

\begin{frontmatter}

\title{A linear algorithm for Brick Wang tiling}

\author[A. Derouet-Jourdan]{Alexandre Derouet-Jourdan}
\ead{alexandre.derouet-jourdan@olm.co.jp}
\address[A. Derouet-Jourdan]{OLM Digital Inc.}

\author[S. Kaji]{Shizuo Kaji\corref{cor1}}
\ead{skaji@yamaguchi-u.ac.jp}
\address[S. Kaji]{Yamaguchi University / JST PRESTO}
\cortext[cor1]{Corresponding author}

\author[Y. Mizoguchi]{Yoshihiro Mizoguchi}
\ead{ym@imi.kyushu-u.ac.jp}
\address[Y. Mizoguchi]{Kyushu University}

\begin{abstract}
The {\em Wang tiling} is a classical problem in combinatorics.
A major theoretical question is to find a (small) set of tiles which
tiles the plane only aperiodically.
In this case, resulting tilings are rather restrictive.
On the other hand, Wang tiles are used as a tool to generate textures and patterns in computer graphics. 
In these applications, a set of tiles is normally chosen so that
it tiles the plane or its sub-regions easily in many different ways.
With computer graphics applications in mind, 
we introduce a class of such tileset, which we call \emph{sequentially permissive} tilesets,
and consider tiling problems with {\em constrained boundary}.
We apply our methodology to a special set of Wang tiles, called Brick Wang tiles,
 introduced by Derouet-Jourdan et al. in 2015 to model wall patterns.
We generalise their result by providing a linear algorithm to
decide and solve the tiling problem for arbitrary planar regions with holes.
\end{abstract}


\end{frontmatter}

\section{Introduction}
Wang tiles are a class of formal systems introduced by Hao Wang in 1961~\cite{wang1961}.
A Wang prototile can be thought of as a square tile with a colour on each side.
Once a set of such prototiles is given, 
copies of prototiles are placed side by side so that the colours of common edges match.
Given a finite set of Wang prototiles,
the {\em domino problem} asks whether a whole plane can be tiled with copies of them,
and was proved undecidable by R.~Berger in 1966~\cite{berger1966}.
There are also sets of Wang prototiles which can tile the plane but only aperiodically
(\cite{berger1966,kari1996,culik1996,jeandel2015}).

In practical applications, bounded planar regions are usually considered,
with which the corresponding domino problem is obviously decidable.
Moreover, sets of Wang prototiles are normally restricted to a permissive class so that 
there exist many solutions in general.
In this case, it is more reasonable to ask for efficient algorithms
to decide tilability with constrained boundary and give solutions if they exist.
We make clear our framework in \S \ref{problem-setting}.
In general, the problem is known to be NP-complete \cite{Le},
so we restrict ourselves to a certain sub-problem.
Namely, we deal with a specific setting which arises in computer graphics applications
such as procedural synthesis of textures, height fields, and sampling points~\cite{cohen2003,kopf2006,stam1997}.
In this case, sets of Wang prototiles are often {\em sequentially permissive} (Def. \ref{3const-property}),
and we have a general strategy described in Lemma \ref{solvability}.
In \S \ref{Brick-Wang-Tiles},
we specialise to the tiling problems with the {\em Brick Wang tiles} (Def. \ref{def:brick-tiles}) introduced in \cite{jourdan2015}
to generate wall pattern textures.
It is proved in \cite{jourdan2015} that any rectangle region greater than the $2\times 2$ square with any boundary colour condition can be tiled. This result was then formally verified and implemented in~\cite{matsushima2016} using the Coq proof assistant~\cite{Coq}.
We generalise the result by giving a complete characterisation of
 the set of regions which can be tiled with any given boundary constraints.
In \S \ref{Solver},
we provide a linear time algorithm for the Brick Wang tiling problem,
with an implementation with source codes at \cite{impl}.

%%%
\section{Problem setting}\label{problem-setting}
We begin with introducing a type of 
Wang tiling problem
which can be regarded as an {\em extension problem}
or a {\em tiling problem with boundary constraints}.
Then, we consider a class of Wang prototiles having a
special property \eqref{3const-property}
and discuss a strategy for a tiling algorithm
(Lemma \ref{solvability}).

We first define the object to be tiled as an enriched graph.
%\red{(Let's restrict ourselves to a finite graphs.
%We can do with infinite graphs, but it requires 
%a definition of tiling e.g., using a radius of a subgraph.
%In this case, we have to do breadth first search instead of depth first.)}
\begin{dfn}
A {\em board} 
is a finite undirected graph $G=(V\cup \{v_0\},E)$
with a distinguished vertex $v_0$ called the 
{\em constrainer} which satisfies the following:
\begin{itemize}
%    \item $V$ is at most countable
%    \item the degree (valency) of any vertex except for $v_0$ is finite 
    \item at each $v\in V$, the set of edges $N(v)$ 
    adjacent to $v$ is ordered: $N(v)=\{e_1(v),\ldots,e_{\deg(v)}(v)\}$, where $\deg(v)$ is the degree of $v$
    \item the full sub-graph on $V$ is connected.
\end{itemize}
Vertices in $V$ are called {\em cells}.
We often identify a board with its underlying graph
when there is no confusion.
\end{dfn}

Figure~\ref{fig:graphtile} shows an intuitive correspondence
between our graphical formalisation and the usual Wang tile.
\begin{figure}[ht!]
\center
\begin{minipage}{0.55\textwidth}
 \includegraphics[width=\textwidth]{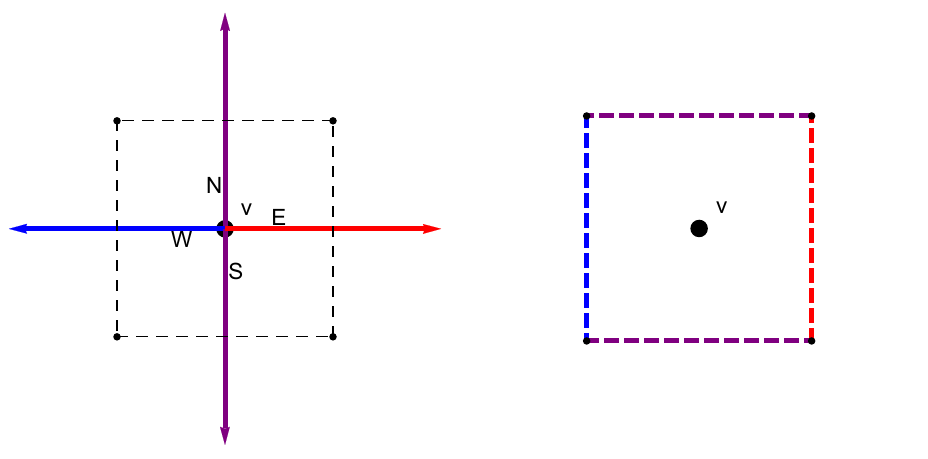}
\caption{Graph vs Wang tile}\label{fig:graphtile}
\end{minipage} 
\end{figure}

\begin{dfn}
A {\em tiling problem} $(G,C,W)$ consists of 
\begin{itemize}
\item a board $G=(V\cup\{v_0\},E)$
\item a set of colours $C$
\item a set of {\em prototiles} $W=\bigcup_{k=1}^\infty W_k$, where $W_k\subset C^k$.
\end{itemize}
A {\em tiling} of a subset $F\subset E$
is a map $\tau:F \to C$ satisfying 
$(\tau(e_1(v)), \ldots, \tau(e_{\deg(v)}(v))) \in W$ for any $v\in V$ with $N(v)\subset F$.
It is said to be {\em full} if $F=E$.
A tiling $\tau':F'\to C$ is called an {\em extension}
of another tiling $\tau:F\to C$ with $F\subset F'$
if $\tau(e)=\tau'(e)$ for any $e\in F$.
\end{dfn}
We are particularly interested in 
finding a full extension $\tau: E\to C$ of
a given tiling
$\beta: N(v_0)\to C$,
which we regard as {\em boundary constraints}.
Edges in $N(v_0)$ are called 
{\em constrained legs} and those in $E\setminus N(v_0)$
are called {\em free legs}.
Given boundary constraints $\beta: N(v_0)\to C$,
the tiling problem with $\beta$ is said to be
\emph{solved} if we find a full extension of $\beta$.
A tiling problem is said to be {\em always solvable}
when there exists a full extension $\tau: E\to C$ of
any tiling $\beta: N(v_0)\to C$.
Given a tiling $\tau:F\to C$,
we informally say to {\em tile} $v\in V$
with $w\in W$
when we can extend $\tau$ by assigning 
$\tau(e_i(v))=w_i\ (1\le i\le \deg(v))$.

\begin{rem}
This formulation encompasses tiling of non-planar graphs such as a surface in the three dimensional space (\cite{Chi}).
For example, a periodic tiling of a planar grid can be regarded as a tiling of a grid graph on the torus.
\end{rem}

\begin{ex}\label{ex:first}
For $n$ and $m$ be natural numbers, we consider a rectangular region
\begin{eqnarray*}
V&=&\{v_{i,j}\,|\,1\le j \le n, 1\le i \le m\},\mbox{\ and}\\
E&=&\{(v_{i,j},v_{i+1,j})\,|\,0\le i \le m, 1\le j \le n\}
\cup \{(v_{i,j},v_{i,j+1})\,|\,1\le i \le m, 0\le j \le n\},
\end{eqnarray*}
where we set $v_{i,j}=v_0$ for $i=0$, $i=m+1$, $j=0$, or $j=n+1$.
For each $v_{i,j} \in V$, the set of edges $N(v_{i,j})$
adjacent to $v_{i,j}$
consists of four edges $e_E(v_{i,j})=(v_{i,j},v_{i+1,j})$,
$e_N(v_{i,j})=(v_{i,j},v_{i,j+1})$,
$e_W(v_{i,j})=(v_{i,j},v_{i-1,j})$,
and $e_S(v_{i,j})=(v_{i,j},v_{i,j-1})$.
The order of $N(v_{i,j})$ is given
by $e_E$, $e_N$, $e_W$, and $e_S$.
These data define a board $G_{m,n}=(V\cup \{v_0\},E)$.

Let $C=\{Red,Blue,Purple,Green\}$,
 $W_E$ be as in Figure~\ref{fig:we},
% $W_E=\{$
% $(Red$, $Red$, $Red$, $Red)$,
% $(Red$, $Red$, $Green$, $Purple)$,
% $(Red$, $Purple$, $Purple$, $Red)$,
% $(Red$, $Purple$, $Blue$, $Purple)$,
% $(Purple$, $Green$, $Red$, $Red)$,
% $(Purple$, $Green$, $Green$, $Purple)$,
% $(Purple$, $Blue$, $Purple$, $Red)$,
% $(Purple$, $Blue$, $Blue$, $Purple)$,
% $(Green$, $Red$, $Red$, $Green)$,
% $(Green$, $Red$, $Green$, $Blue)$,
% $(Green$, $Purple$, $Purple$, $Green)$,
% $(Green$, $Purple$, $Blue$, $Blue)$,
% $(Blue$, $Green$, $Red$, $Green)$,
% $(Blue$, $Green$, $Green$, $Blue)$,
% $(Blue$, $Blue$, $Purple$, $Green)$,
% $(Blue$, $Blue$, $Blue$, $Blue)\}$
and consider the tiling problem $(G_{5,3}, C, W_E)$.
% It is trivial that this problem is always solvable.
We set boundary constraints $\beta: N(v_0)\to C$
 as in Figure~\ref{fig:boundary}.
% that is $\beta(e_W(v_{1,1}))=Blue$,
% $\beta(e_S(v_{1,1}))=Purple$,
% $\beta(e_S(v_{2,1}))=Green$,
% $\beta(e_S(v_{3,1}))=Green$,
% $\beta(e_S(v_{4,1}))=Blue$, and so on.
We note that edges connected to $v_0$ are denoted by arrows
in the Figure.
Also we consider double edges between
$v_0$ and $v_{1,1}$, $v_{5,1}$, $v_{1,3}$ and $v_{5,3}$
which may have different colours.
An example of full extension of $\beta$ is given by
$\tau:E \to C$ defined as in Figure~\ref{fig:tiling}.
%that is, $e_E(v_{1,1})=Purple$, $e_N(v_{1,1})=Blue$,
%$e_W(v_{1,1})=Blue$, $e_S(v_{1,1})=Purple$, and so on.

\begin{figure}
\begin{minipage}{0.35\textwidth}
 \includegraphics[width=\textwidth]{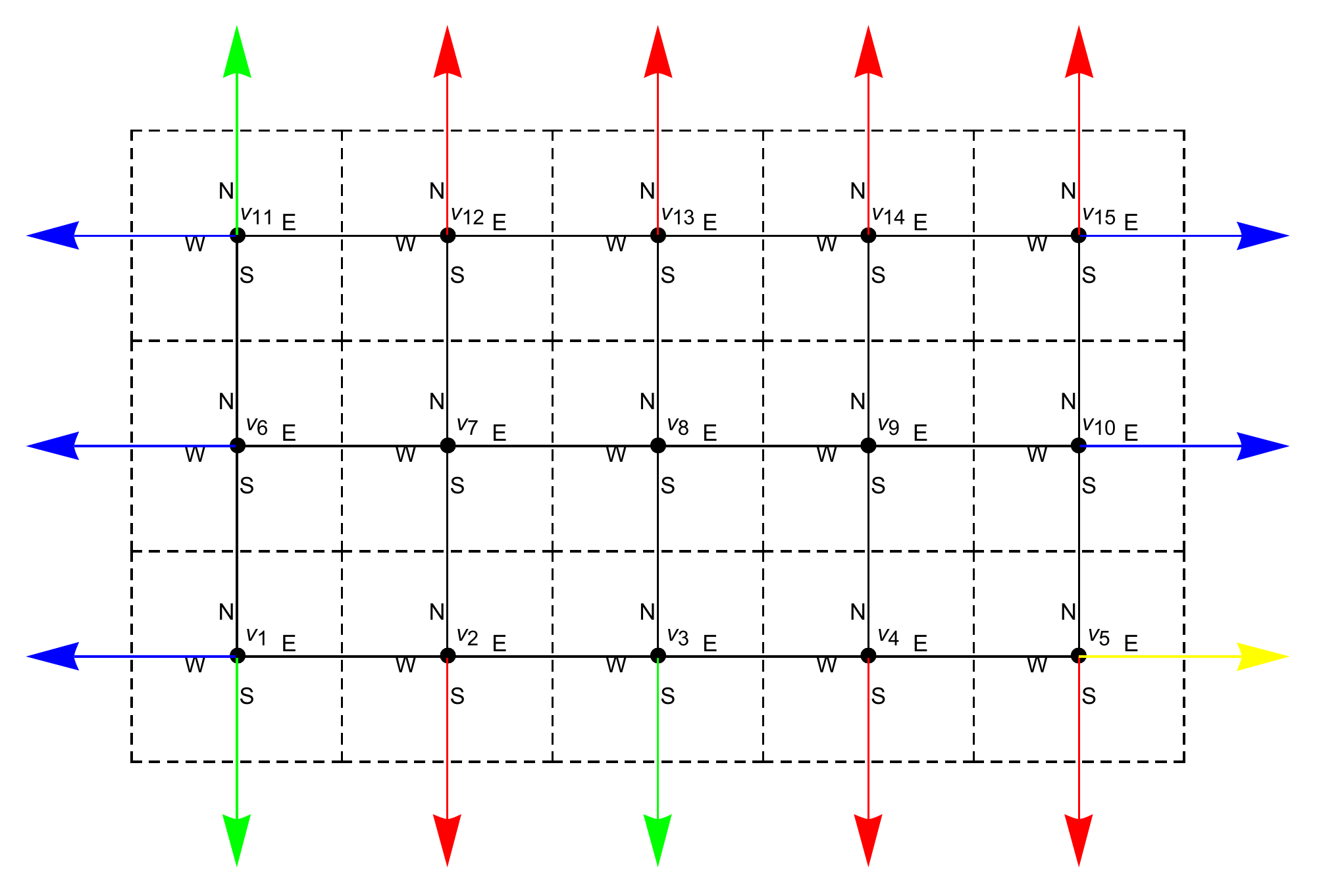}
\caption{Boundary constraints}\label{fig:boundary}
\end{minipage} 
\begin{minipage}{0.35\linewidth}
 \includegraphics[width=\textwidth]{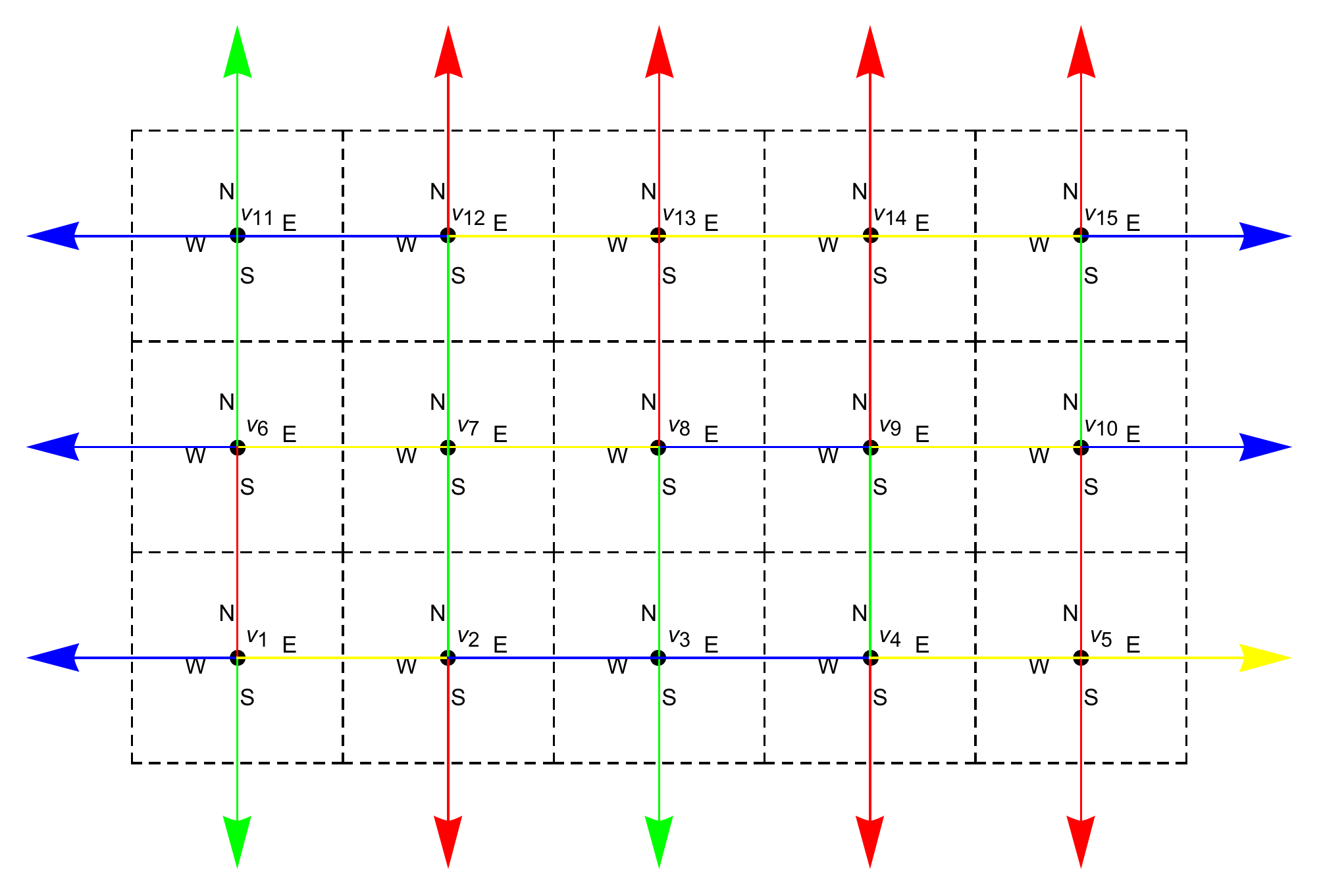}
\caption{A full extension}\label{fig:tiling}
\end{minipage} 
\begin{minipage}{0.25\linewidth}
\includegraphics[width=\textwidth]{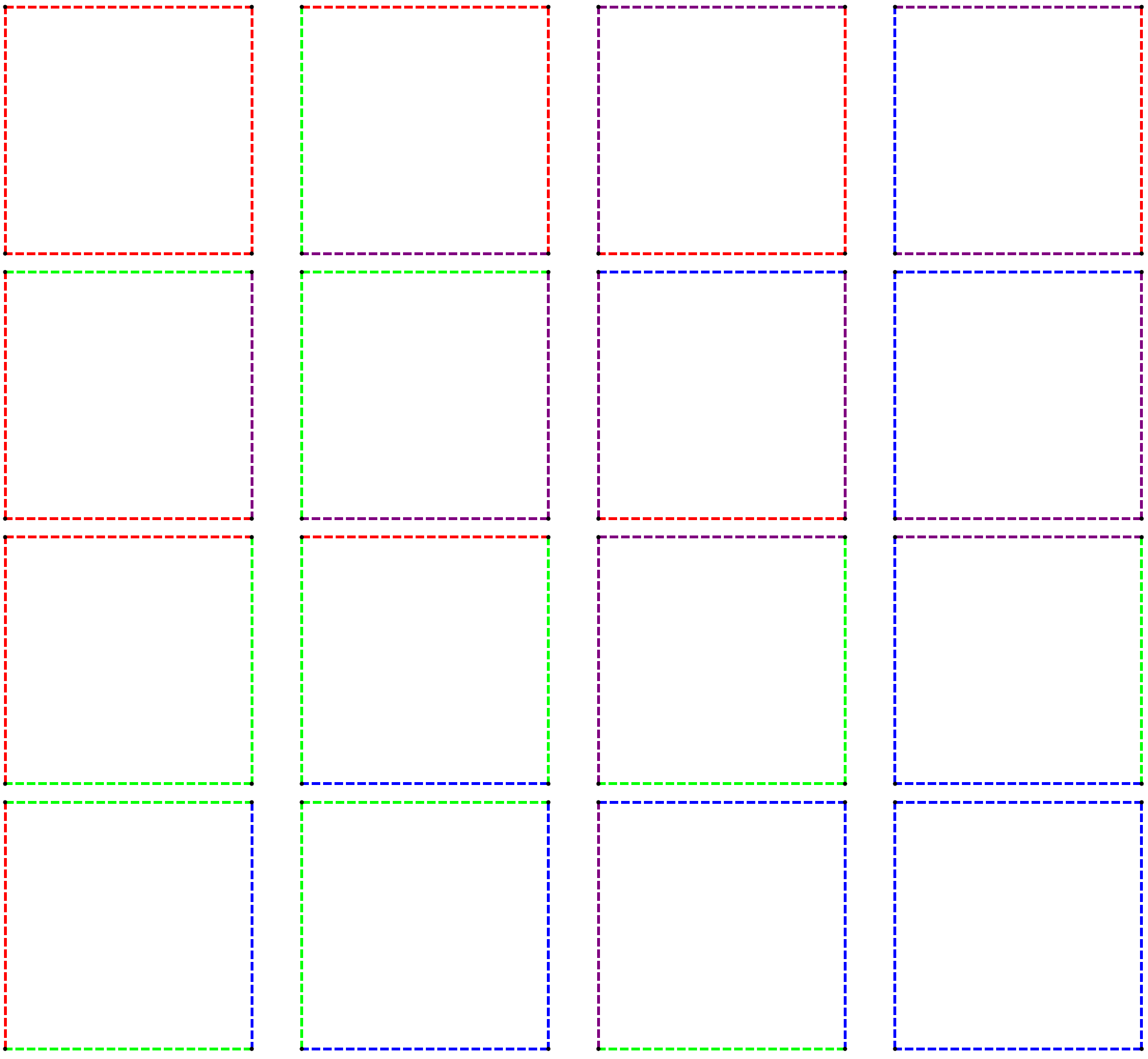}
\caption{Wang prototiles in $W_E$}\label{fig:we}
\end{minipage} 
\end{figure}

\end{ex}

\begin{ex}\label{ex:corner}
Here we give an example of a board which is not a simple graph.
For natural numbers $n$ and $m$, we consider again a rectangular region
\begin{eqnarray*}
V&=&\{v_{i,j}\,|\,1\le j \le n, 1\le i \le m\}.
\end{eqnarray*}
This time however the edge set is given as follows:
for each $v_{i,j}\in V$, 
the set of adjacent edges is given by
\begin{eqnarray*}
N(v_{i,j})&=&\{
e_{ENE},e_{NE},e_{NNE},
e_{NNW},e_{SE},e_{WNW},
e_{WSW},e_{SW},e_{SSW},
e_{SSE},e_{NW},e_{ESE}\},
\end{eqnarray*}
where the target of each edge is respectively
\[
v_{i+1,j}, ,v_{i+1,j+1}, ,v_{i,j+1}, v_{i,j+1},
v_{i-1,j+1}, v_{i-1,j}, v_{i-1,j}, v_{i-1,j-1},
v_{i,j-1}, v_{i,j-1}, v_{i+1,j-1}, v_{i+1,j}.
\]
(See Figure~\ref{fig:cornertile}.)
Here our convention is $v_{i,j}=v_0$ for $i=0$, $i=m+1$, $j=0$ or $j=n+1$.
These data define a board $G'_{m,n}=(V\cup \{v_0\},E=\bigcup_{v\in V}N(v))$.
The {\bf corner Wang tiling} introduced in \cite{lagae2006}
is a special instance of tiling problem $(G'_{m,n},C,W_C)$, where
$W_C=\{(c_i)\in C^{12} \mid c_{3k+1}=c_{3k+2}=c_{3k+3}, k=0,1,2,3\}$.
In Figure~\ref{fig:cornertile}, We show an example of a correspondence
between our graphical formalisation and the corner Wang tile in
\cite{lagae2006}.

For example, let $C=\{Red,Blue\}$,
and consider a tiling problem $(G'_{2,3},C,W_C)$.
Figure~\ref{fig:cornertiling} gives a solution to the problem.
\begin{figure}[ht]
\begin{center}
\begin{minipage}{0.45\textwidth}
 \includegraphics[width=\textwidth]{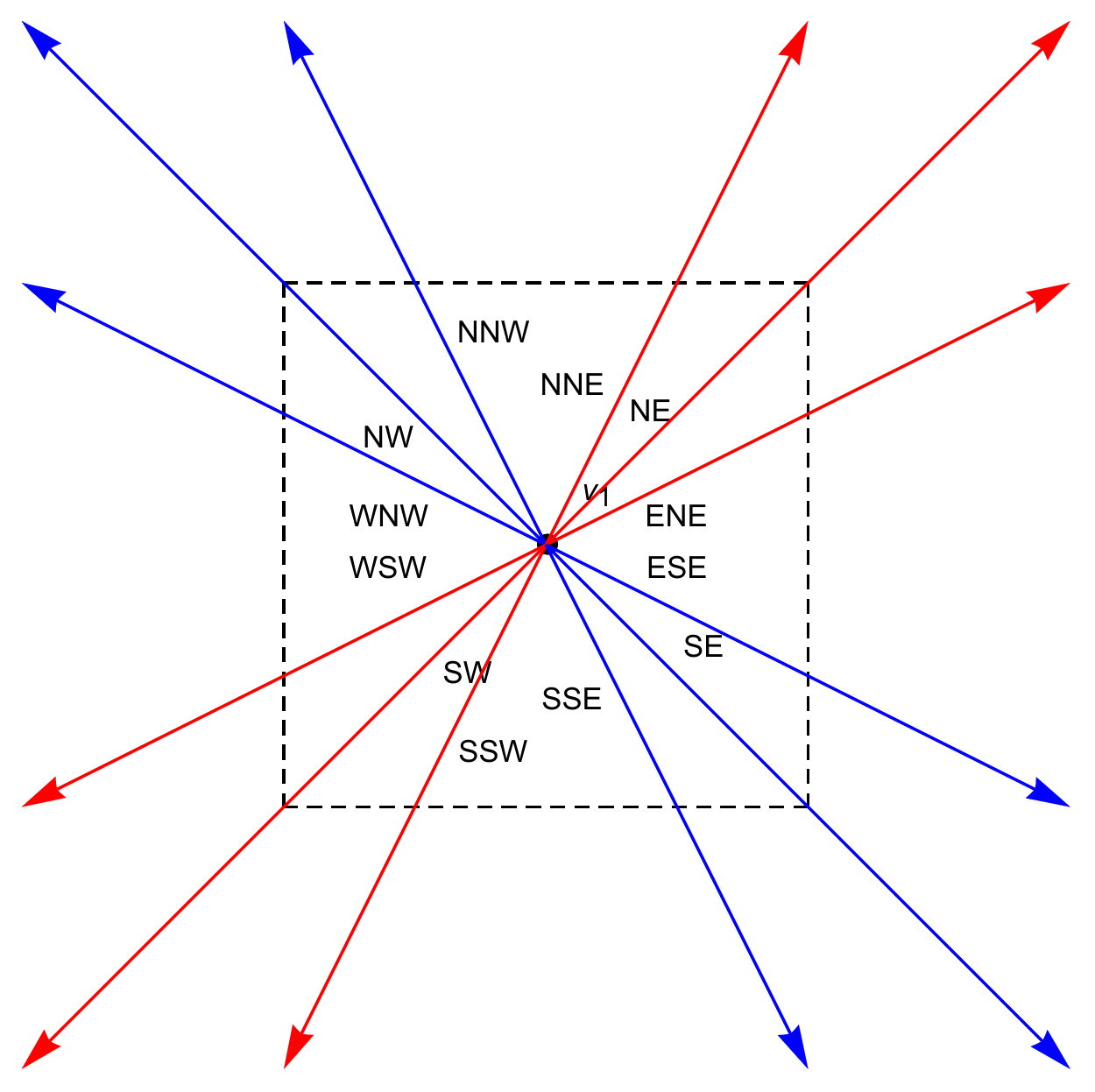}
\caption{Corner Wang tile}\label{fig:cornertile}
\end{minipage} 
\begin{minipage}{0.45\linewidth}
 \includegraphics[width=\textwidth]{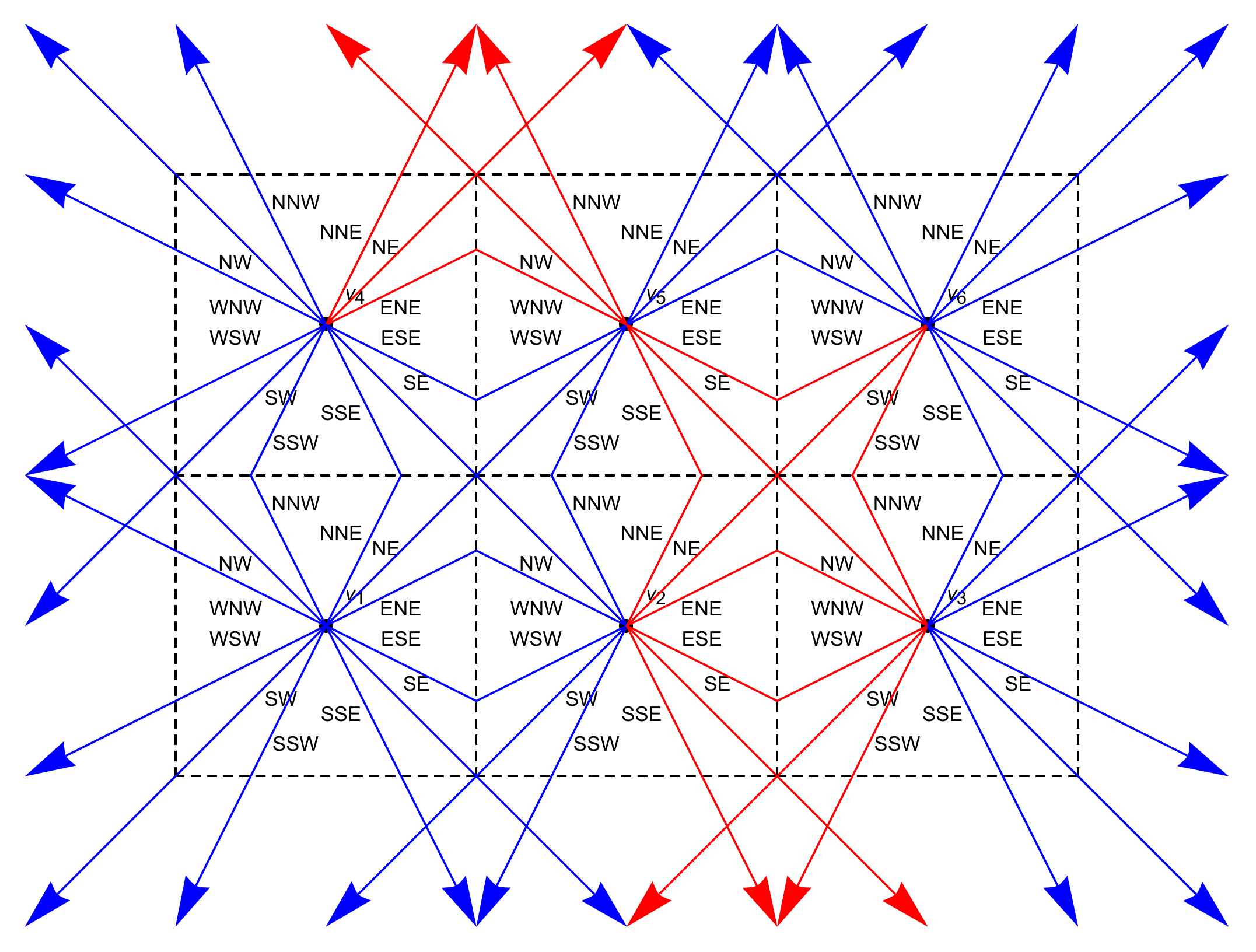}
\caption{Corner Wang tiling}\label{fig:cornertiling}
\end{minipage} 
\end{center}
\end{figure}

As mentioned in \cite{lagae2006},
a corner Wang tiling is equivalent to the ordinary Wang tiling with more colours.
Let $\tilde{C}=\{Red,Blue,Purple,Green\}$.
We think of a purple edge having blue and red ends,
and a green edge having red and blue ends.
The tiling problem $(G_{2,3},\tilde{C},W_E)$ defined in Example~\ref{ex:first}
is equivalent to $(G'_{2,3},C,W_C)$.
A solution of $(G_{2,3},\tilde{C},W_E)$ in Figure~\ref{fig:cornertiling}
corresponds to a solution of $(G'_{2,3},C,W_C)$
in Figure~\ref{fig:equivalenttiling}.
In Figure~\ref{fig:edgetiling}, we display the solution
using usual square Wang tiles and corner Wang tiles.

\begin{figure}[ht]
\begin{center}
\begin{minipage}{0.3\linewidth}
\includegraphics[width=\textwidth]{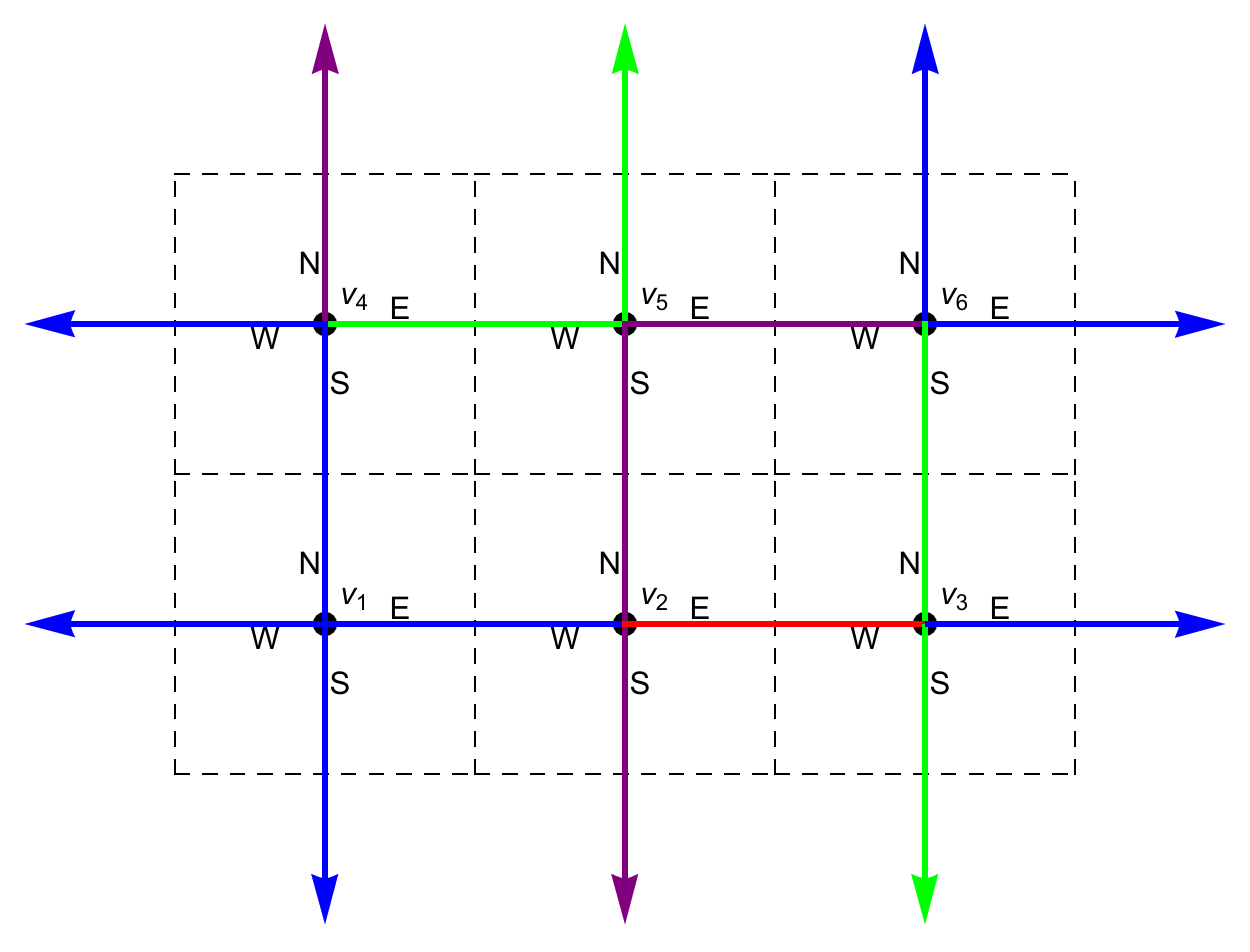}
\caption{Equivalent Wang tiling}\label{fig:equivalenttiling}
\end{minipage} 
\begin{minipage}{0.65\linewidth}
\includegraphics[width=\textwidth]{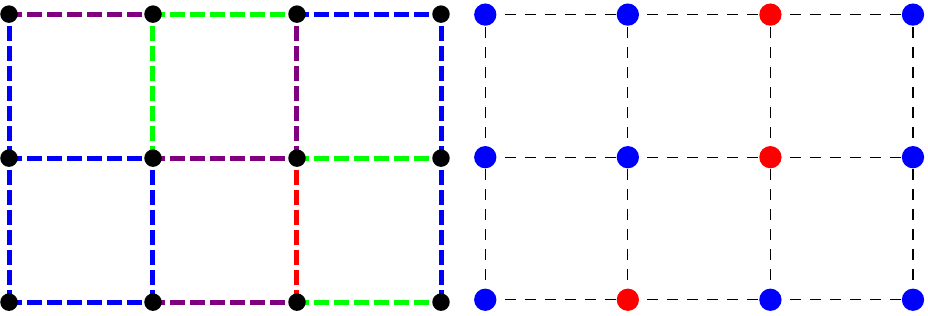}
\caption{Edge Wang tiling and Corner Wang tiling}\label{fig:edgetiling}
\end{minipage} 
\end{center}
\end{figure}

\end{ex}

\begin{dfn}
For a board $G=(V\cup\{v_0\},E)$ and its sub-graph $(U,F)$,
where $U\subset V$,
the {\em restriction} $G|_{(U,F)}$ of $G$ is 
a board obtained by contracting the complement of $(U,F)$ to the constrainer $v_0$:
The vertex set is the union $U \cup \{v_0\}$
and the edge set is
\[
E|_{(U,F)}=F \cup \{(u,v_0) \mid u \in U, (u,v) \in E\setminus F\}.
\]
That is, an edge not belonging to $F$ is cut into two pieces and each of them is reconnected to $v_0$.
For a full sub-graph of $G$ on $U\subset V$,
we abbreviate the corresponding restriction as $G|_U$.
\end{dfn}

In many applications of Wang tiling in computer graphics,
the set of Wang prototiles often comes equipped with a nice property,
which is rather permissive in terms of tilability.
\begin{dfn}
For a natural number $k$, 
a set $W$ of prototiles is said to be {\em $k$-sequentially permissive} 
if the following holds:
\begin{align}\label{3const-property}
\pi_{\bar{i}}: W_{k} \to C^{k-1} \text{ is surjective for all } 1\le i\le k,
\end{align}
where $\pi_{\bar{i}}$ is the projection dropping the $i$-th factor.
In other words, a cell with $k$-legs 
can always be tiled when it has at least one free leg.
\end{dfn}

\begin{ex}
Figure~\ref{fig:example} shows an example, taken from~\cite{cohen2003}, of a set of prototiles
in computer graphics which is $4$-sequentially permissive.
In practical applications, the colours
    are replaced by parts of an image to generate textures. Essentially, each side of an edge is one half of an image cut on the diagonal. That way, when tiles are placed adjacently, the transition from one tile to another in the picture is smooth. When all colours are taken from pieces of a given image, and the Wang tiles stitch together the colours, it produces arbitrary large textures from a small example. In~\cite{cohen2003}, this technique is applied to create large Poisson distribution of points by replacing the image parts with carefully designed point distributions.
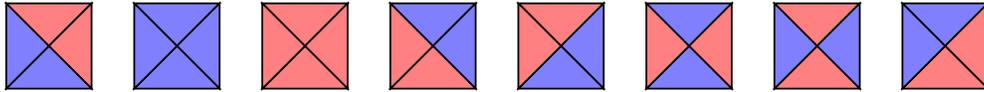
\begin{figure}[ht]
    \centering
    \resizebox{0.75\linewidth}{!}{
	\begin{tikzpicture}
		\begin{scope}
			\filldraw[color=black, fill=red!50, very thick] (0,0) -- (2,0) -- (1,-1) -- cycle;
			\filldraw[color=black, fill=red!50, very thick] (2,0) -- (2,-2) -- (1,-1) -- cycle;
			\filldraw[color=black, fill=blue!50, very thick] (2,-2) -- (0,-2) -- (1,-1) -- cycle;
			\filldraw[color=black, fill=blue!50, very thick] (0,-2) -- (0,0) -- (1,-1) -- cycle;
		\end{scope}
		\begin{scope}[xshift=3cm]
			\filldraw[color=black, fill=blue!50, very thick] (0,0) -- (2,0) -- (1,-1) -- cycle;
			\filldraw[color=black, fill=blue!50, very thick] (2,0) -- (2,-2) -- (1,-1) -- cycle;
			\filldraw[color=black, fill=blue!50, very thick] (2,-2) -- (0,-2) -- (1,-1) -- cycle;
			\filldraw[color=black, fill=blue!50, very thick] (0,-2) -- (0,0) -- (1,-1) -- cycle;
		\end{scope}
		\begin{scope}[xshift=6cm]
			\filldraw[color=black, fill=red!50, very thick] (0,0) -- (2,0) -- (1,-1) -- cycle;
			\filldraw[color=black, fill=red!50, very thick] (2,0) -- (2,-2) -- (1,-1) -- cycle;
			\filldraw[color=black, fill=red!50, very thick] (2,-2) -- (0,-2) -- (1,-1) -- cycle;
			\filldraw[color=black, fill=red!50, very thick] (0,-2) -- (0,0) -- (1,-1) -- cycle;
		\end{scope}
		\begin{scope}[xshift=9cm]
			\filldraw[color=black, fill=blue!50, very thick] (0,0) -- (2,0) -- (1,-1) -- cycle;
			\filldraw[color=black, fill=blue!50, very thick] (2,0) -- (2,-2) -- (1,-1) -- cycle;
			\filldraw[color=black, fill=red!50, very thick] (2,-2) -- (0,-2) -- (1,-1) -- cycle;
			\filldraw[color=black, fill=red!50, very thick] (0,-2) -- (0,0) -- (1,-1) -- cycle;
		\end{scope}
		\begin{scope}[xshift=12cm]
			\filldraw[color=black, fill=red!50, very thick] (0,0) -- (2,0) -- (1,-1) -- cycle;
			\filldraw[color=black, fill=blue!50, very thick] (2,0) -- (2,-2) -- (1,-1) -- cycle;
			\filldraw[color=black, fill=blue!50, very thick] (2,-2) -- (0,-2) -- (1,-1) -- cycle;
			\filldraw[color=black, fill=red!50, very thick] (0,-2) -- (0,0) -- (1,-1) -- cycle;
		\end{scope}
		\begin{scope}[xshift=15cm]
			\filldraw[color=black, fill=blue!50, very thick] (0,0) -- (2,0) -- (1,-1) -- cycle;
			\filldraw[color=black, fill=red!50, very thick] (2,0) -- (2,-2) -- (1,-1) -- cycle;
			\filldraw[color=black, fill=blue!50, very thick] (2,-2) -- (0,-2) -- (1,-1) -- cycle;
			\filldraw[color=black, fill=red!50, very thick] (0,-2) -- (0,0) -- (1,-1) -- cycle;
		\end{scope}
		\begin{scope}[xshift=18cm]
			\filldraw[color=black, fill=red!50, very thick] (0,0) -- (2,0) -- (1,-1) -- cycle;
			\filldraw[color=black, fill=blue!50, very thick] (2,0) -- (2,-2) -- (1,-1) -- cycle;
			\filldraw[color=black, fill=red!50, very thick] (2,-2) -- (0,-2) -- (1,-1) -- cycle;
			\filldraw[color=black, fill=blue!50, very thick] (0,-2) -- (0,0) -- (1,-1) -- cycle;
		\end{scope}
		\begin{scope}[xshift=21cm]
			\filldraw[color=black, fill=blue!50, very thick] (0,0) -- (2,0) -- (1,-1) -- cycle;
			\filldraw[color=black, fill=red!50, very thick] (2,0) -- (2,-2) -- (1,-1) -- cycle;
			\filldraw[color=black, fill=red!50, very thick] (2,-2) -- (0,-2) -- (1,-1) -- cycle;
			\filldraw[color=black, fill=blue!50, very thick] (0,-2) -- (0,0) -- (1,-1) -- cycle;
		\end{scope}
	\end{tikzpicture}
}
    \caption{Example of a set of Wang tiles used in computer graphics~\cite{cohen2003}.}
    \label{fig:example}
\end{figure}
\end{ex}

% Sequentially permissive prototiles can tile any
% infinite board:
% \begin{lem}
% Let $W$ be a set of prototiles which is sequentially permissive (see \eqref{3const-property}).
% A tiling problem $(G,C,W)$ is always solvable when
% $V$ is infinite.
% \end{lem}
% \begin{proof}
% Construct a spanning tree rooted at $v_0$.
% Pick a leaf and tile it by \eqref{3const-property}.
% If there is no leaf, split it into two parts.
% They are both infinite.
% \end{proof}

The following lemma provides a general reduction strategy
to solve tiling problems:
\begin{lem}\label{solvability}
Let $G$ be a board with the set of cells $V$.
Let $W$ be a set of prototiles which is $k$-sequentially permissive \eqref{3const-property}
for all $k\in \{\deg(v)\mid v\in V\}$.
A tiling problem $(G,C,W)$ is always solvable when
there exists a restriction $G|_{U}$ such that $(G|_{U},C,W)$ is always solvable.
\end{lem}
\begin{proof}
Pick a cell $v\in V$ which is adjacent to $U$,
and consider $(G|_{U\cup\{v\}},C,W)$.
The vertex $v$ has at least one free leg, which is connected to $U$.
So we can tile $v$ by \eqref{3const-property} for any $\beta:N(v_0)\to C$.
Now we are left with $(G|_{U},C,W)$ with extended constraints $\bar{\beta}$,
which can be fully extended by assumption.
Since $V$ is finite, the assertion follows by induction.
\end{proof}
In practice, we proceed in the ``opposite'' way;
we first construct a spanning tree
of $G|_{V\setminus U}$ and tile it from leaves by
\eqref{3const-property}, and then finally 
cells in $U$ 
are tiled by assumption (see \S \ref{Solver}).

%%%%%%%%%%%%%%%%%%%%%%
\section{Brick Wang tiles}\label{Brick-Wang-Tiles}
%%%%%%%%%%%%%%%%%%%%%%
From now on, we focus on a special set of prototiles 
which we call the {\em brick Wang tiles}.
We give a necessary and sufficient condition for a 
tiling problem to be always solvable.

The brick Wang tiles were introduced in \cite{jourdan2015} to model wall patterns.
Each tile represents how corners of four bricks meet. 
It is assumed that the edges of the bricks are axis aligned and that each tile is traversed with a straight line, either vertically or horizontally. For aesthetic reasons, crosses are forbidden, where all four bricks are aligned and the corresponding tile is traversed by two straight lines. 
In this model, the colour in the Wang tiles indicates the position of the edge of the brick on the side of the tile. 
The formal definition is given as follows:
\begin{dfn}\label{def:brick-tiles}
We always assume $\# C\ge 3$.
The set of {\em brick Wang prototiles} $W_B$ is defined by
\[
W_B=\{ (c_1,c_2,c_3,c_4)\in C^4 \mid (c_1=c_3 \wedge c_2\neq c_4) \text{ or } (c_1 \neq c_3 \wedge c_2= c_4)\}.
\]
\end{dfn}
It is easy to see that $W_B$ is
$4$-sequentially permissive \eqref{3const-property}.

\begin{figure}[htb]
\centering
\resizebox{0.7\width}{0.7\height}{\begin{tikzpicture}
[
    point/.style = {draw, circle,  fill = black, inner sep = 1pt},
]

\newcommand{\pythagwidth}{3cm}
\newcommand{\pythagheight}{2cm}
\newcommand{\tilesize}{3cm}

\coordinate [] (A) at (0, 0);
\coordinate [] (B) at (0, \tilesize);
\coordinate [] (C) at (\tilesize, \tilesize);
\coordinate [] (D) at (\tilesize, 0);

\draw [dashed] (A) -- (B) -- (C) -- (D) -- (A);

\coordinate [] (E) at (0.75 * \tilesize, \tilesize);
\coordinate [] (F) at (\tilesize, 0.5 * \tilesize);
\coordinate [] (G) at (0.25 * \tilesize, 0);
\coordinate [] (H) at (0, 0.5 * \tilesize);

\coordinate [] (E') at (0.75 * \tilesize, 0.5 * \tilesize);
\coordinate [] (G') at (0.25 * \tilesize, 0.5 * \tilesize);

\draw [very thick] (H) -- (F);
\draw [very thick] (G) -- (G');
\draw [very thick] (E) -- (E');

\node at (H) [point] {};
\node at (F) [point] {};
\node at (E) [point] {};
\node at (G) [point] {};

%\draw[-latex,thick](-2,0.7 * \tilesize)node[above]
%        {Wang tile edge} to[out=-90,in=180] (-0.2, 0.25 * \tilesize);

%\draw[-latex,thick](2.5,-0.7 )node[right]
%        {Brick edge} to[out=180,in=-90] (0.5 * \tilesize, 0.45 * \tilesize);
		
\draw[-latex,thick](-0.5*\tilesize, 0.6* \tilesize)node[left]
        {edge position $c_1$} to[out=-90,in=180] (0, 0.5 * \tilesize);

\draw[-latex,thick](1.1*\tilesize,1.2* \tilesize)node[right]
        {edge position $c_4$} to[out=-180,in=90] (0.75 * \tilesize, 1.05 * \tilesize);

\draw[-latex,thick](1.2*\tilesize, 0.6* \tilesize)node[right]
        {edge position $c_3$} to[out=-90,in=0] (\tilesize, 0.5 * \tilesize);

\draw[-latex,thick](-0.5*\tilesize,-0.7)node[left]
        {edge position $c_2$} to[out=0,in=-90] (0.25 * \tilesize, 0);

\end{tikzpicture}}
\includegraphics[width=0.2\linewidth]{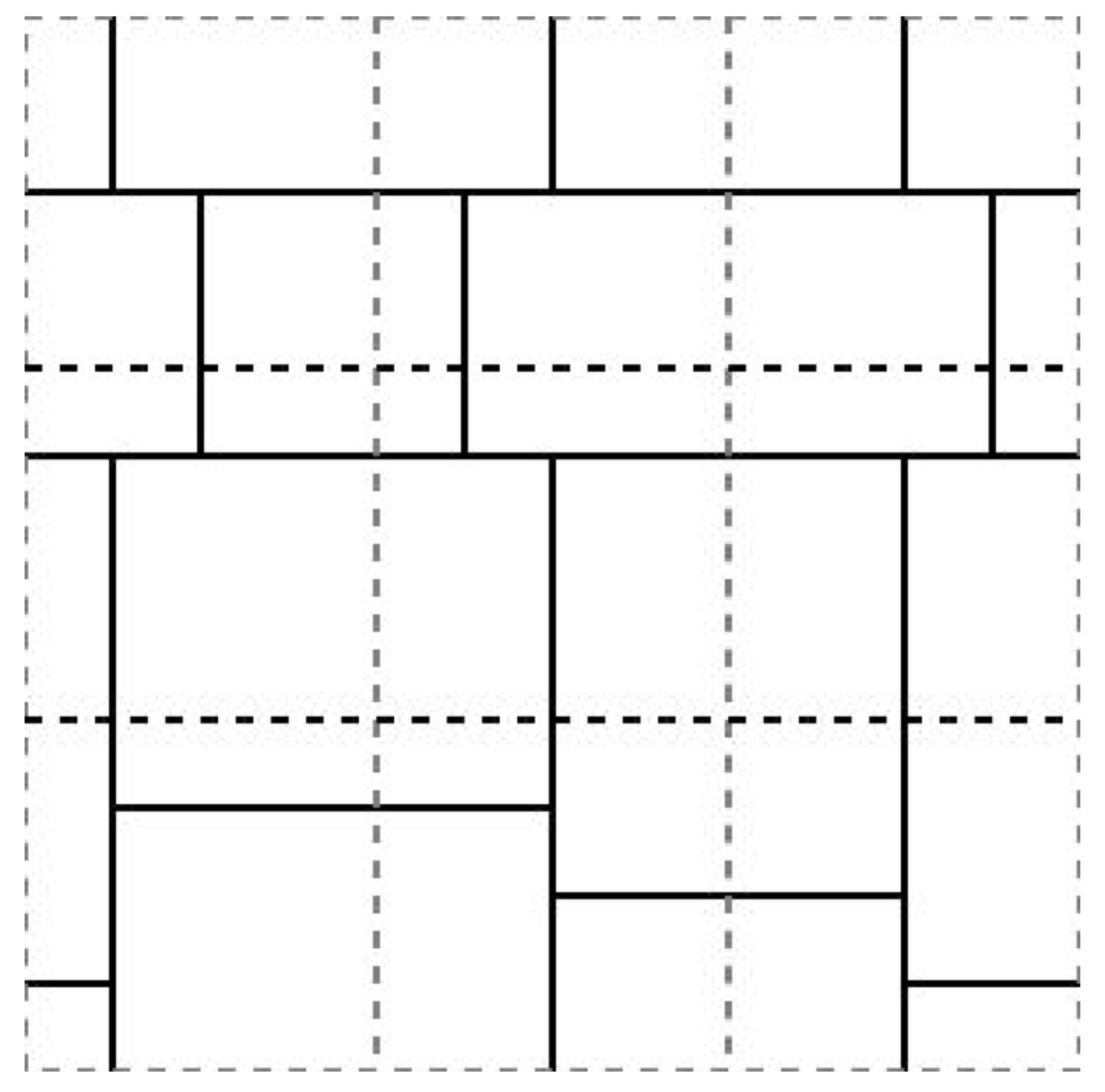}
\caption{A brick Wang tile and a tiling of $3\times 3$-board}\label{fig:brickwangtile}
\end{figure}
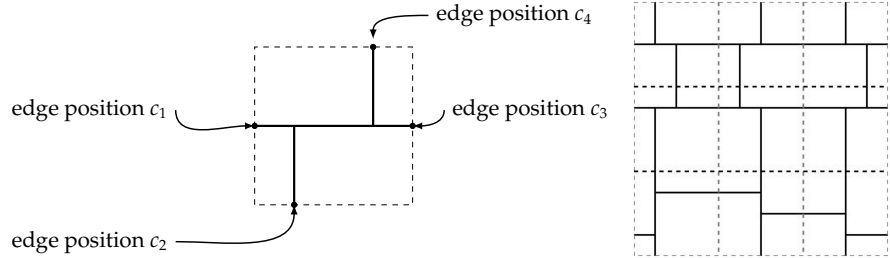

In the previous work \cite{jourdan2015,matsushima2016},
only rectangular regions are considered.
Here, we deal with any planar region possibly with holes.
In practice, this allows the user to manually tile a part of the region,
and leave the rest to be automatically completed by our algorithm in \S \ref{Solver}. See Example~\ref{ex:result}.

\begin{dfn}
Consider the graph $\bar{G}=(V\cup\{v_0\},E)$, where  
\begin{itemize}
    \item $V= \{ (i,j)\in \Z\times \Z \}$
    \item $E=\{ ((i,j),(i',j'))\in V\times V\mid |i-i'|+|j-j'|=1 \}$.
\end{itemize}
The set $N((i,j))$ of edges adjacent to a vertex $(i,j)$ are given ordering by 
\[
 ( ((i,j),(i+1,j)), ((i,j),(i,j+1)), ((i,j),(i-1,j)), ((i,j),(i,j-1)).
\]
A board is said to be \emph{square grid} if it is a restriction of $\bar{G}$ to a finite sub-graph.
\end{dfn}

We give a simple necessary and sufficient condition for a square grid board $G$ to be always solvable with $W_B$.
A {\em path} is an ordered sequence of distinct cells 
$(v_1,v_2,\ldots,v_l)$ where consecutive cells are connected by an edge.
A {\em cycle} is an ordered sequence of cells
$(v_1,v_2,\ldots,v_{l-1},v_l=v_1)$ where all $v_i$ are distinct, and consecutive cells are connected by an edge.
A cell $v_i$ in a path or a cycle is said to be \emph{straight} if edges $e_j,e_k\in N(v_i)$
 connected to consecutive cells satisfy $|j-k|=2$ with the ordering of $N(v_i)$.
{\em Corner} cells in a path or a cycle are those that are not straight.
%An {\em $\infty$-trail} in $G$ is a trail $(v_1,e_1,v_2,e_2,\ldots,e_{l-1}v_l)$
%where $v_1$ coincides with some $v_j$ with $j\neq 1$ and $v_l$ coincides with some $v_k$ with $k\neq l$.

%A {\em straight cell} in a trail is a cell which is not a corner.

\begin{lem}\label{lem:base}
A tiling problem $(G,C,W_B)$ for a square grid board $G$ 
is always solvable
when the full sub-graph
on the cells is a cycle.
\end{lem}
\begin{proof}
Let $\beta: N(v_0)\to C$ be any colouring of the constrained legs,
and $(v_1,\ldots,v_l=v_1)$ be the full sub-graph of $G$ on the cells.
For a cell $v_k$, we say the edge connected to $v_{k-1}$ the incoming leg, denoted by $in(v_k)$,
and the edge connected to $v_{k+1}$ the outgoing leg, denoted by $out(v_k)$.

First, we can assume there is
no straight cells whose constrained legs are
given two different colours
since they have no effect.
Denote by $L_{a,b}$ a corner cell whose constrained 
leg opposite to the incoming leg is coloured with $a\in C$ 
and the outgoing leg with $b\in C$.
Also, denote by $S$ a straight cell whose constrained legs are coloured with the same colour.
These serve as a kind of logical gate: 
\begin{itemize}
\item $\begin{cases} \tau(in(C_{a,b}))=a \Rightarrow \tau(out(C_{a,b}))\neq b \\
 \tau(in(C_{a,b}))\neq a \Rightarrow \tau(out(C_{a,b}))= b \end{cases}$
\item $in(S)\neq out(S)$.
\end{itemize}
We prove the assertion by case by case analysis.
Note that we can change the ordering of $(v_1,\ldots,v_l=v_1)$ cyclically, or reverse it.
\begin{itemize}
\item Assume that $v_1=S, v_2=C_{a,b}$.
Set $\tau(out(C_{a,b}))=b$ and tile $(v_3,v_4,\ldots,v_{l-1})$ sequentially by \eqref{3const-property}.
If $out(v_{l-1})$ is coloured with $c$, choose a colour for $in(S)$ which is different from $c$.
By choosing a colour for $out(S)$ which is different from both $c$ and $a$, we obtain a solution.
\item From the previous case, we can now assume that there is no straight cell.
Assume that $v_1=C_{a,b}, v_2=C_{c,d}$ with $b\neq c$.
Set $\tau(out(C_{c,d}))=d$ and tile $(v_3,v_4,\ldots,v_{l-1})$ sequentially by \eqref{3const-property}.
If $\tau(out(v_{l-1}))=a$, 
choose a colour for $out(C_{a,b})$ which are different from both $b$ and $c$.
If $\tau(out(v_{l-1}))\neq a$, 
set $\tau(out(C_{a,b}))= b$.
This gives a solution.
\item
Now, the only remaining case is when
$(v_1,v_2,\ldots,v_{l-1})=(C_{a_1,a_2},C_{a_2,a_3},\cdots C_{a_{l-1},a_1})$.
Note that $l-1$ should be even since the board is a square grid.
Set $\tau(in(C_{a_i,a_{i+1}}))=a_i$
and $\tau(out(C_{a_i,a_{i+1}}))\neq a_{i+1}$
for odd $i$.
This gives a solution.
\end{itemize}
\end{proof}

Combining Lemma \ref{solvability} with Lemma \ref{lem:base}, we obtain
\begin{thm}
\label{thm:plan_solv}
A tiling problem $(G,C,W_B)$ for a square grid board $G$ is always solvable
if and only if the full sub-graph
on the cells contains a cycle.
\end{thm}
\begin{proof}
We only have to show the ``only if'' part.
Given any board $G$ whose full sub-graph on $V$ is a tree,
we construct $\beta: N(v_0)\to C$ which cannot be fully extended.
We proceed by induction on the number of cells $\# V$.
Pick a leaf $v\in V$ and set $G'=G_{V\setminus \{v\}}$. 
Then, by induction, 
there exists $\beta': N'(v_0)\to C$
which cannot be fully extended in $G'$, where $N'(v_0)$ is the edges of $G'$ adjacent to $v_0$.
Let $u$ be the cell connected to $v$ with the edge $e$.
We can choose the value of $\beta$ for $N(v)\setminus \{e\}$ so that
the colour of $e$ is forced to coincide with $\beta'(e)$ in order to tile $v$.
Then, the resulting $\beta$ cannot be fully extended.
\end{proof}

\begin{rem}
Note that a $2\times 2$ region discussed in \cite{jourdan2015}
is a cycle of length four in our language.
So the above theorem is a generalisation of the result in \cite{jourdan2015}.
\end{rem}

%%%%%%%%%%%%%%%%%%%%%%
\section{The tiling algorithm for brick Wang tiles}\label{Solver}
%%%%%%%%%%%%%%%%%%%%%%
In this section, we will discuss an algorithm to decide 
what $\beta$ can be extended
to a full tiling $\tau:E\to C$ and how we can obtain one.
An implementation is given at \cite{impl}.

Fix a tiling problem with 
a finite planar square grid board $G$ with the Brick Wang tiles $W_B$ and a number of colours $\# C\ge 3$. 
Fix also boundary constraints $\beta: N(v_0)\to C$.
Based on Theorem~\ref{thm:plan_solv}, we design an algorithm to 
decide solubility and to construct a valid tiling 
of the problem.
The outline goes as follows: let $G'$ be the full sub-graph on the cells
\begin{enumerate}
    \item Detect a cycle in $G'$ by performing a depth first traversal of $G'$.
    If there is no cycle, use the tree solver given in Lemma~\ref{lem:tree}.
    \item Pick a cell $v$ which is adjacent to but not in the cycle.
    Construct a tiling in a depth-first manner; 
    that is, tile sequentially from the leaves of a spanning tree rooted at $v$ of the connected component of the complement of the cycle. 
    \item Repeat Step 2 until there is no cell other than those in the cycle.
    \item Construct a tiling of the cycle using the procedure of Lemma~\ref{lem:base}.
\end{enumerate}
Since the degree of a cell is always four, then the complexity of the cycle extraction is linear in the number of cells in $G$.
Similarly, we will see that all other operations are performed in a linear time with regards to the number of cells in the board.

The remaining task is now to decide and solve a tree-structured board. To this extent, we introduce the \emph{condition propagation} that will tell how the border colouring put constraints on the inner legs.

A map $f:C \to \{true, false\}$ of one of the following forms is called a {\em condition}:
\begin{itemize}
	\item $(c)$ It must match a given colour $c$; $f(x)=true \Leftrightarrow x=c$.
	\item $(\neg c)$ It must not match a given colour $c$; $f(x)=true \Leftrightarrow x\neq c$.
	\item $(*)$ It can match any colour; $f(x)=true$ for any $x$.
\end{itemize}
To each free leg $e$, we assign a condition 
$f_e$ in such a way that a map $\tau: E\to C$
satisfying $\tau(x)=\beta(x)$ for $x\in N(v_0)$ and
$f_e(\tau(e))=true$ for all free legs $e$ 
gives a full extension of $\beta$.

Fix a cell $v$. If for every free leg of $v$ except for one
is assigned a condition,
a condition for the last free leg is canonically determined by choosing the weakest one (lower in the above list means weaker).
We say that the condition is inferred on the last leg.
The precise rule
of this is described in Figure~\ref{fig:cond_prop_gen}. 
For a tree rooted at $v$, we can propagate the condition recursively 
by inferring from the leafs to the root $v$.
\begin{lem}\label{lem:tree}
Let $G$ be a board whose full sub-graph on the cells is a tree rooted at $v$.
For a tiling problem $(G,C,W_B)$, 
boundary constraints $\beta: N(v_0)\to C$ admit a full extension if and only 
if $v$ can be tiled in such a way that it
matches propagated conditions on its legs.
Furthermore, any such tiling of $v$ can be extended to a full extension.
\end{lem}
\begin{proof}
Only if part is trivial by definition.

Observe that the inferred conditions are 
so that for any choice of a matching colour for a free leg of a cell, there exists at least one colouring of 
all other free legs which gives a valid tiling of the cell.

Tile $v$ in such a way that it
matches propagated conditions on its legs.
Consider the connected components obtained by removing $v$ from $G$, with boundary constraints induced by $\beta$ and 
the tiling of $v$.
Now the proof goes by induction on the number of cells.
\end{proof}
The algorithm is illustrated in Figure~\ref{fig:tree_alg}.
Since each cell is visited at most four times, it has a linear complexity with respect to the number of cells.

\begin{figure}[ht]
    \centering
    \resizebox{0.8\linewidth}{!}{
    \begin{tikzpicture}
			%\draw[step=2,very thin, dashed] (-1.000000,-7.000000) grid (11.000000,1.000000);
			\begin{scope}
				\draw[very thick] (0,-2) -- (2,-2) node [midway, below] {$(a)$};
				\draw[very thick] (2,-2) -- (2,-4) node [midway, left] (start) {$\cdot$};
				\draw[very thick] (2,-4) -- (0,-4) node [midway, above] {$(b)$};
				\draw[very thick] (0,-4) -- (0,-2) node [midway, right] {$\mathcal{C}$};
				\draw[very thick] (0,-2) -- (2,-4);
				\draw[very thick] (2,-2) -- (0,-4);
				
				\node[anchor=west] at (4, -2.2) (case1) {$\overline{\mathcal{C}}$ if $a = b$};
				\node[anchor=west] at (4, -3.8) (case2) {$\mathcal{C}$ if $a \neq b$};
				\draw (start) edge[out=0,in=180,->] (case1) ;
				\draw (start) edge[out=0,in=180,->] (case2) ;
			\end{scope}
			
			\begin{scope}[xshift=8cm]
				\draw[very thick] (0,-2) -- (2,-2) node [midway, below] {$(\neg a)$};
				\draw[very thick] (2,-2) -- (2,-4) node [midway, left] (start) {$\cdot$};
				\draw[very thick] (2,-4) -- (0,-4) node [midway, above] {$(b)$};
				\draw[very thick] (0,-4) -- (0,-2) node [midway, right] {$\mathcal{C}$};
				\draw[very thick] (0,-2) -- (2,-4);
				\draw[very thick] (2,-2) -- (0,-4);
				
				\node[anchor=west] at (4, -2.2) (case1) {$\mathcal{C}$ if $a = b$};
				\node[anchor=west] at (4, -3.8) (case2) {$(*)$ if $a \neq b$};
				\draw (start) edge[out=0,in=180,->] (case1) ;
				\draw (start) edge[out=0,in=180,->] (case2) ;
			\end{scope}
			\begin{scope}[xshift=16cm]
				\draw[very thick] (0,-2) -- (2,-2) node [midway, below] {$(\neg a)$};
				\draw[very thick] (2,-2) -- (2,-4) node [midway, left]  {$(*)$};
				\draw[very thick] (2,-4) -- (0,-4) node [midway, above] {$(\neg b)$};
				\draw[very thick] (0,-4) -- (0,-2) node [midway, right] {$\mathcal{C}$};
				\draw[very thick] (0,-2) -- (2,-4);
				\draw[very thick] (2,-2) -- (0,-4);
			\end{scope}

		\end{tikzpicture}
    
    }
    \caption{Propagation of conditions for cells in a tree. When some of the legs of the cell are on the boundary, we consider them input legs with a condition $(c)$ where $c$ is the colour of the leg. $\mathcal{C}$ denotes a given condition and $\overline{\mathcal{C}}$ its negation, with $\overline{(a)}=(\neg a)$, $\overline{(\neg a)} = (*)$ and $\overline{(*)}=(*)$. When any of input conditions is $(*)$, then the output is $(*)$.}
    \label{fig:cond_prop_gen}
\end{figure}
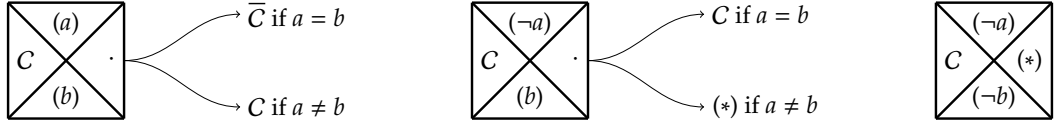

%%Then, by propagating the condition from the input legs of $v_1$ to its output and then from cell to cell, we obtain a condition on the output leg of $v_l$. If the output colour satisfies this condition, then the path is solvable and we can construct a valid tiling by resolving the conditions from the output to the input.

%%\subsection{Case of a tree-structured board}\label{tree-solver}
%%Using the condition propagation, it is possible to decide whether 
%%a tiling problem with a tree structured board is solvable.
%%Consider the border colouring of the board as conditions on the boundary legs.
%%In a bottom-up manner starting from leaves, 
%%compute the output condition of each cell using the bordering colouring. Then, propagate conditions up to the root of the tree and check that the conditions are valid at the root. If not, the problem is not solvable. If it is solvable, 
%%a valid tiling is given by tiling from the root
%%to meet the conditions.

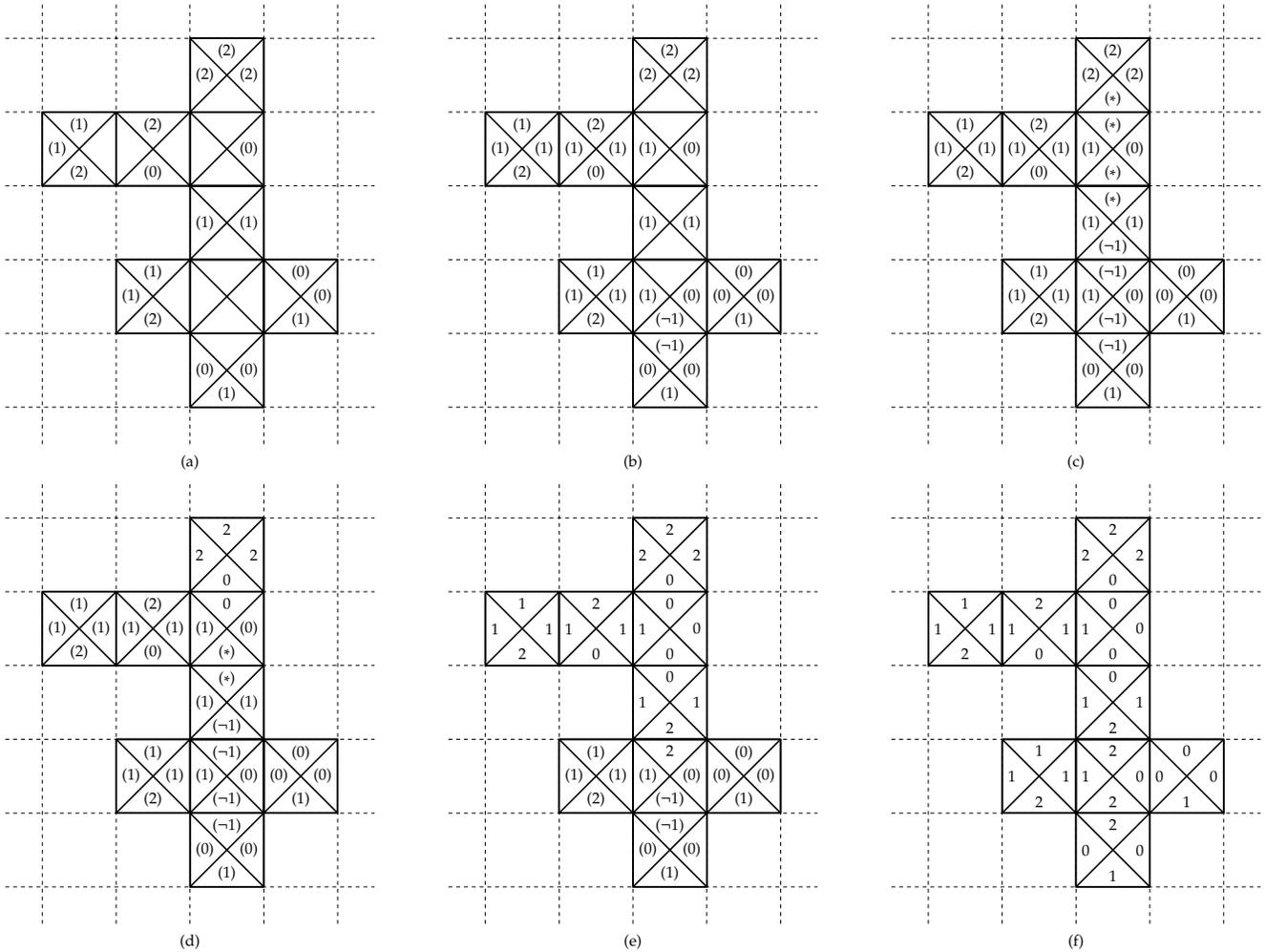
\begin{figure}[ht]
    \centering
    \resizebox{\linewidth}{!}{
	\begin{tikzpicture}
		\begin{scope}
			\draw[step=2,very thin, dashed] (-1.000000,-11.000000) grid (9.000000,1.000000);
            \draw[very thick] (4,0) -- (6,0) node [midway, below] {$(2)$};
            \draw[very thick] (6,0) -- (6,-2) node [midway, left] {$(2)$};
            \draw[very thick] (6,-2) -- (4,-2) node [midway, above] {};
            \draw[very thick] (4,-2) -- (4,0) node [midway, right] {$(2)$};
            \draw[very thick] (4,0) -- (6,-2);
            \draw[very thick] (6,0) -- (4,-2);
            \draw[very thick] (4,0) -- (6,0) -- (6,-2) -- (4,-2) -- cycle;
            \draw[very thick] (0,-2) -- (2,-2) node [midway, below] {$(1)$};
            \draw[very thick] (2,-2) -- (2,-4) node [midway, left] {};
            \draw[very thick] (2,-4) -- (0,-4) node [midway, above] {$(2)$};
            \draw[very thick] (0,-4) -- (0,-2) node [midway, right] {$(1)$};
            \draw[very thick] (0,-2) -- (2,-4);
            \draw[very thick] (2,-2) -- (0,-4);
            \draw[very thick] (0,-2) -- (2,-2) -- (2,-4) -- (0,-4) -- cycle;
            \draw[very thick] (2,-2) -- (4,-2) node [midway, below] {$(2)$};
            \draw[very thick] (4,-2) -- (4,-4) node [midway, left] {};
            \draw[very thick] (4,-4) -- (2,-4) node [midway, above] {$(0)$};
            \draw[very thick] (2,-4) -- (2,-2) node [midway, right] {};
            \draw[very thick] (2,-2) -- (4,-4);
            \draw[very thick] (4,-2) -- (2,-4);
            \draw[very thick] (2,-2) -- (4,-2) -- (4,-4) -- (2,-4) -- cycle;
            \draw[very thick] (4,-2) -- (6,-2) node [midway, below] {};
            \draw[very thick] (6,-2) -- (6,-4) node [midway, left] {$(0)$};
            \draw[very thick] (6,-4) -- (4,-4) node [midway, above] {};
            \draw[very thick] (4,-4) -- (4,-2) node [midway, right] {};
            \draw[very thick] (4,-2) -- (6,-4);
            \draw[very thick] (6,-2) -- (4,-4);
            \draw[very thick] (4,-2) -- (6,-2) -- (6,-4) -- (4,-4) -- cycle;
            \draw[very thick] (4,-4) -- (6,-4) node [midway, below] {};
            \draw[very thick] (6,-4) -- (6,-6) node [midway, left] {$(1)$};
            \draw[very thick] (6,-6) -- (4,-6) node [midway, above] {};
            \draw[very thick] (4,-6) -- (4,-4) node [midway, right] {$(1)$};
            \draw[very thick] (4,-4) -- (6,-6);
            \draw[very thick] (6,-4) -- (4,-6);
            \draw[very thick] (4,-4) -- (6,-4) -- (6,-6) -- (4,-6) -- cycle;
            \draw[very thick] (2,-6) -- (4,-6) node [midway, below] {$(1)$};
            \draw[very thick] (4,-6) -- (4,-8) node [midway, left] {};
            \draw[very thick] (4,-8) -- (2,-8) node [midway, above] {$(2)$};
            \draw[very thick] (2,-8) -- (2,-6) node [midway, right] {$(1)$};
            \draw[very thick] (2,-6) -- (4,-8);
            \draw[very thick] (4,-6) -- (2,-8);
            \draw[very thick] (2,-6) -- (4,-6) -- (4,-8) -- (2,-8) -- cycle;
            \draw[very thick] (4,-6) -- (6,-6) node [midway, below] {};
            \draw[very thick] (6,-6) -- (6,-8) node [midway, left] {};
            \draw[very thick] (6,-8) -- (4,-8) node [midway, above] {};
            \draw[very thick] (4,-8) -- (4,-6) node [midway, right] {};
            \draw[very thick] (4,-6) -- (6,-8);
            \draw[very thick] (6,-6) -- (4,-8);
            \draw[very thick] (4,-6) -- (6,-6) -- (6,-8) -- (4,-8) -- cycle;
            \draw[very thick] (6,-6) -- (8,-6) node [midway, below] {$(0)$};
            \draw[very thick] (8,-6) -- (8,-8) node [midway, left] {$(0)$};
            \draw[very thick] (8,-8) -- (6,-8) node [midway, above] {$(1)$};
            \draw[very thick] (6,-8) -- (6,-6) node [midway, right] {};
            \draw[very thick] (6,-6) -- (8,-8);
            \draw[very thick] (8,-6) -- (6,-8);
            \draw[very thick] (6,-6) -- (8,-6) -- (8,-8) -- (6,-8) -- cycle;
            \draw[very thick] (4,-8) -- (6,-8) node [midway, below] {};
            \draw[very thick] (6,-8) -- (6,-10) node [midway, left] {$(0)$};
            \draw[very thick] (6,-10) -- (4,-10) node [midway, above] {$(1)$};
            \draw[very thick] (4,-10) -- (4,-8) node [midway, right] {$(0)$};
            \draw[very thick] (4,-8) -- (6,-10);
            \draw[very thick] (6,-8) -- (4,-10);
            \draw[very thick] (4,-8) -- (6,-8) -- (6,-10) -- (4,-10) -- cycle;
            \node[anchor=center] at (4, -11.5) {(a)};
		\end{scope}
		\begin{scope}[xshift=12cm]
			\draw[step=2,very thin, dashed] (-1.000000,-11.000000) grid (9.000000,1.000000);
            \draw[very thick] (4,0) -- (6,0) node [midway, below] {$(2)$};
            \draw[very thick] (6,0) -- (6,-2) node [midway, left] {$(2)$};
            \draw[very thick] (6,-2) -- (4,-2) node [midway, above] {};
            \draw[very thick] (4,-2) -- (4,0) node [midway, right] {$(2)$};
            \draw[very thick] (4,0) -- (6,-2);
            \draw[very thick] (6,0) -- (4,-2);
            \draw[very thick] (4,0) -- (6,0) -- (6,-2) -- (4,-2) -- cycle;
            \draw[very thick] (0,-2) -- (2,-2) node [midway, below] {$(1)$};
            \draw[very thick] (2,-2) -- (2,-4) node [midway, left] {$(1)$};
            \draw[very thick] (2,-4) -- (0,-4) node [midway, above] {$(2)$};
            \draw[very thick] (0,-4) -- (0,-2) node [midway, right] {$(1)$};
            \draw[very thick] (0,-2) -- (2,-4);
            \draw[very thick] (2,-2) -- (0,-4);
            \draw[very thick] (0,-2) -- (2,-2) -- (2,-4) -- (0,-4) -- cycle;
            \draw[very thick] (2,-2) -- (4,-2) node [midway, below] {$(2)$};
            \draw[very thick] (4,-2) -- (4,-4) node [midway, left] {$(1)$};
            \draw[very thick] (4,-4) -- (2,-4) node [midway, above] {$(0)$};
            \draw[very thick] (2,-4) -- (2,-2) node [midway, right] {$(1)$};
            \draw[very thick] (2,-2) -- (4,-4);
            \draw[very thick] (4,-2) -- (2,-4);
            \draw[very thick] (2,-2) -- (4,-2) -- (4,-4) -- (2,-4) -- cycle;
            \draw[very thick] (4,-2) -- (6,-2) node [midway, below] {};
            \draw[very thick] (6,-2) -- (6,-4) node [midway, left] {$(0)$};
            \draw[very thick] (6,-4) -- (4,-4) node [midway, above] {};
            \draw[very thick] (4,-4) -- (4,-2) node [midway, right] {$(1)$};
            \draw[very thick] (4,-2) -- (6,-4);
            \draw[very thick] (6,-2) -- (4,-4);
            \draw[very thick] (4,-2) -- (6,-2) -- (6,-4) -- (4,-4) -- cycle;
            \draw[very thick] (4,-4) -- (6,-4) node [midway, below] {};
            \draw[very thick] (6,-4) -- (6,-6) node [midway, left] {$(1)$};
            \draw[very thick] (6,-6) -- (4,-6) node [midway, above] {};
            \draw[very thick] (4,-6) -- (4,-4) node [midway, right] {$(1)$};
            \draw[very thick] (4,-4) -- (6,-6);
            \draw[very thick] (6,-4) -- (4,-6);
            \draw[very thick] (4,-4) -- (6,-4) -- (6,-6) -- (4,-6) -- cycle;
            \draw[very thick] (2,-6) -- (4,-6) node [midway, below] {$(1)$};
            \draw[very thick] (4,-6) -- (4,-8) node [midway, left] {$(1)$};
            \draw[very thick] (4,-8) -- (2,-8) node [midway, above] {$(2)$};
            \draw[very thick] (2,-8) -- (2,-6) node [midway, right] {$(1)$};
            \draw[very thick] (2,-6) -- (4,-8);
            \draw[very thick] (4,-6) -- (2,-8);
            \draw[very thick] (2,-6) -- (4,-6) -- (4,-8) -- (2,-8) -- cycle;
            \draw[very thick] (4,-6) -- (6,-6) node [midway, below] {};
            \draw[very thick] (6,-6) -- (6,-8) node [midway, left] {$(0)$};
            \draw[very thick] (6,-8) -- (4,-8) node [midway, above] {$(\neg 1)$};
            \draw[very thick] (4,-8) -- (4,-6) node [midway, right] {$(1)$};
            \draw[very thick] (4,-6) -- (6,-8);
            \draw[very thick] (6,-6) -- (4,-8);
            \draw[very thick] (4,-6) -- (6,-6) -- (6,-8) -- (4,-8) -- cycle;
            \draw[very thick] (6,-6) -- (8,-6) node [midway, below] {$(0)$};
            \draw[very thick] (8,-6) -- (8,-8) node [midway, left] {$(0)$};
            \draw[very thick] (8,-8) -- (6,-8) node [midway, above] {$(1)$};
            \draw[very thick] (6,-8) -- (6,-6) node [midway, right] {$(0)$};
            \draw[very thick] (6,-6) -- (8,-8);
            \draw[very thick] (8,-6) -- (6,-8);
            \draw[very thick] (6,-6) -- (8,-6) -- (8,-8) -- (6,-8) -- cycle;
            \draw[very thick] (4,-8) -- (6,-8) node [midway, below] {$(\neg 1)$};
            \draw[very thick] (6,-8) -- (6,-10) node [midway, left] {$(0)$};
            \draw[very thick] (6,-10) -- (4,-10) node [midway, above] {$(1)$};
            \draw[very thick] (4,-10) -- (4,-8) node [midway, right] {$(0)$};
            \draw[very thick] (4,-8) -- (6,-10);
            \draw[very thick] (6,-8) -- (4,-10);
            \draw[very thick] (4,-8) -- (6,-8) -- (6,-10) -- (4,-10) -- cycle;
            \node[anchor=center] at (4, -11.5) {(b)};
		\end{scope}
		\begin{scope}[xshift=24cm]
			\draw[step=2,very thin, dashed] (-1.000000,-11.000000) grid (9.000000,1.000000);
            \draw[very thick] (4,0) -- (6,0) node [midway, below] {$(2)$};
            \draw[very thick] (6,0) -- (6,-2) node [midway, left] {$(2)$};
            \draw[very thick] (6,-2) -- (4,-2) node [midway, above] {$(*)$};
            \draw[very thick] (4,-2) -- (4,0) node [midway, right] {$(2)$};
            \draw[very thick] (4,0) -- (6,-2);
            \draw[very thick] (6,0) -- (4,-2);
            \draw[very thick] (4,0) -- (6,0) -- (6,-2) -- (4,-2) -- cycle;
            \draw[very thick] (0,-2) -- (2,-2) node [midway, below] {$(1)$};
            \draw[very thick] (2,-2) -- (2,-4) node [midway, left] {$(1)$};
            \draw[very thick] (2,-4) -- (0,-4) node [midway, above] {$(2)$};
            \draw[very thick] (0,-4) -- (0,-2) node [midway, right] {$(1)$};
            \draw[very thick] (0,-2) -- (2,-4);
            \draw[very thick] (2,-2) -- (0,-4);
            \draw[very thick] (0,-2) -- (2,-2) -- (2,-4) -- (0,-4) -- cycle;
            \draw[very thick] (2,-2) -- (4,-2) node [midway, below] {$(2)$};
            \draw[very thick] (4,-2) -- (4,-4) node [midway, left] {$(1)$};
            \draw[very thick] (4,-4) -- (2,-4) node [midway, above] {$(0)$};
            \draw[very thick] (2,-4) -- (2,-2) node [midway, right] {$(1)$};
            \draw[very thick] (2,-2) -- (4,-4);
            \draw[very thick] (4,-2) -- (2,-4);
            \draw[very thick] (2,-2) -- (4,-2) -- (4,-4) -- (2,-4) -- cycle;
            \draw[very thick] (4,-2) -- (6,-2) node [midway, below] {$(*)$};
            \draw[very thick] (6,-2) -- (6,-4) node [midway, left] {$(0)$};
            \draw[very thick] (6,-4) -- (4,-4) node [midway, above] {$(*)$};
            \draw[very thick] (4,-4) -- (4,-2) node [midway, right] {$(1)$};
            \draw[very thick] (4,-2) -- (6,-4);
            \draw[very thick] (6,-2) -- (4,-4);
            \draw[very thick] (4,-2) -- (6,-2) -- (6,-4) -- (4,-4) -- cycle;
            \draw[very thick] (4,-4) -- (6,-4) node [midway, below] {$(*)$};
            \draw[very thick] (6,-4) -- (6,-6) node [midway, left] {$(1)$};
            \draw[very thick] (6,-6) -- (4,-6) node [midway, above] {$(\neg 1)$};
            \draw[very thick] (4,-6) -- (4,-4) node [midway, right] {$(1)$};
            \draw[very thick] (4,-4) -- (6,-6);
            \draw[very thick] (6,-4) -- (4,-6);
            \draw[very thick] (4,-4) -- (6,-4) -- (6,-6) -- (4,-6) -- cycle;
            \draw[very thick] (2,-6) -- (4,-6) node [midway, below] {$(1)$};
            \draw[very thick] (4,-6) -- (4,-8) node [midway, left] {$(1)$};
            \draw[very thick] (4,-8) -- (2,-8) node [midway, above] {$(2)$};
            \draw[very thick] (2,-8) -- (2,-6) node [midway, right] {$(1)$};
            \draw[very thick] (2,-6) -- (4,-8);
            \draw[very thick] (4,-6) -- (2,-8);
            \draw[very thick] (2,-6) -- (4,-6) -- (4,-8) -- (2,-8) -- cycle;
            \draw[very thick] (4,-6) -- (6,-6) node [midway, below] {$(\neg 1)$};
            \draw[very thick] (6,-6) -- (6,-8) node [midway, left] {$(0)$};
            \draw[very thick] (6,-8) -- (4,-8) node [midway, above] {$(\neg 1)$};
            \draw[very thick] (4,-8) -- (4,-6) node [midway, right] {$(1)$};
            \draw[very thick] (4,-6) -- (6,-8);
            \draw[very thick] (6,-6) -- (4,-8);
            \draw[very thick] (4,-6) -- (6,-6) -- (6,-8) -- (4,-8) -- cycle;
            \draw[very thick] (6,-6) -- (8,-6) node [midway, below] {$(0)$};
            \draw[very thick] (8,-6) -- (8,-8) node [midway, left] {$(0)$};
            \draw[very thick] (8,-8) -- (6,-8) node [midway, above] {$(1)$};
            \draw[very thick] (6,-8) -- (6,-6) node [midway, right] {$(0)$};
            \draw[very thick] (6,-6) -- (8,-8);
            \draw[very thick] (8,-6) -- (6,-8);
            \draw[very thick] (6,-6) -- (8,-6) -- (8,-8) -- (6,-8) -- cycle;
            \draw[very thick] (4,-8) -- (6,-8) node [midway, below] {$(\neg 1)$};
            \draw[very thick] (6,-8) -- (6,-10) node [midway, left] {$(0)$};
            \draw[very thick] (6,-10) -- (4,-10) node [midway, above] {$(1)$};
            \draw[very thick] (4,-10) -- (4,-8) node [midway, right] {$(0)$};
            \draw[very thick] (4,-8) -- (6,-10);
            \draw[very thick] (6,-8) -- (4,-10);
            \draw[very thick] (4,-8) -- (6,-8) -- (6,-10) -- (4,-10) -- cycle;
            \node[anchor=center] at (4, -11.5) {(c)};
		\end{scope}
		\begin{scope}[yshift=-13cm]
			\draw[step=2,very thin, dashed] (-1.000000,-11.000000) grid (9.000000,1.000000);
            \draw[very thick] (4,0) -- (6,0) node [midway, below] {$2$};
            \draw[very thick] (6,0) -- (6,-2) node [midway, left] {$2$};
            \draw[very thick] (6,-2) -- (4,-2) node [midway, above] {$0$};
            \draw[very thick] (4,-2) -- (4,0) node [midway, right] {$2$};
            \draw[very thick] (4,0) -- (6,-2);
            \draw[very thick] (6,0) -- (4,-2);
            \draw[very thick] (4,0) -- (6,0) -- (6,-2) -- (4,-2) -- cycle;
            \draw[very thick] (0,-2) -- (2,-2) node [midway, below] {$(1)$};
            \draw[very thick] (2,-2) -- (2,-4) node [midway, left] {$(1)$};
            \draw[very thick] (2,-4) -- (0,-4) node [midway, above] {$(2)$};
            \draw[very thick] (0,-4) -- (0,-2) node [midway, right] {$(1)$};
            \draw[very thick] (0,-2) -- (2,-4);
            \draw[very thick] (2,-2) -- (0,-4);
            \draw[very thick] (0,-2) -- (2,-2) -- (2,-4) -- (0,-4) -- cycle;
            \draw[very thick] (2,-2) -- (4,-2) node [midway, below] {$(2)$};
            \draw[very thick] (4,-2) -- (4,-4) node [midway, left] {$(1)$};
            \draw[very thick] (4,-4) -- (2,-4) node [midway, above] {$(0)$};
            \draw[very thick] (2,-4) -- (2,-2) node [midway, right] {$(1)$};
            \draw[very thick] (2,-2) -- (4,-4);
            \draw[very thick] (4,-2) -- (2,-4);
            \draw[very thick] (2,-2) -- (4,-2) -- (4,-4) -- (2,-4) -- cycle;
            \draw[very thick] (4,-2) -- (6,-2) node [midway, below] {$0$};
            \draw[very thick] (6,-2) -- (6,-4) node [midway, left] {$(0)$};
            \draw[very thick] (6,-4) -- (4,-4) node [midway, above] {$(*)$};
            \draw[very thick] (4,-4) -- (4,-2) node [midway, right] {$(1)$};
            \draw[very thick] (4,-2) -- (6,-4);
            \draw[very thick] (6,-2) -- (4,-4);
            \draw[very thick] (4,-2) -- (6,-2) -- (6,-4) -- (4,-4) -- cycle;
            \draw[very thick] (4,-4) -- (6,-4) node [midway, below] {$(*)$};
            \draw[very thick] (6,-4) -- (6,-6) node [midway, left] {$(1)$};
            \draw[very thick] (6,-6) -- (4,-6) node [midway, above] {$(\neg 1)$};
            \draw[very thick] (4,-6) -- (4,-4) node [midway, right] {$(1)$};
            \draw[very thick] (4,-4) -- (6,-6);
            \draw[very thick] (6,-4) -- (4,-6);
            \draw[very thick] (4,-4) -- (6,-4) -- (6,-6) -- (4,-6) -- cycle;
            \draw[very thick] (2,-6) -- (4,-6) node [midway, below] {$(1)$};
            \draw[very thick] (4,-6) -- (4,-8) node [midway, left] {$(1)$};
            \draw[very thick] (4,-8) -- (2,-8) node [midway, above] {$(2)$};
            \draw[very thick] (2,-8) -- (2,-6) node [midway, right] {$(1)$};
            \draw[very thick] (2,-6) -- (4,-8);
            \draw[very thick] (4,-6) -- (2,-8);
            \draw[very thick] (2,-6) -- (4,-6) -- (4,-8) -- (2,-8) -- cycle;
            \draw[very thick] (4,-6) -- (6,-6) node [midway, below] {$(\neg 1)$};
            \draw[very thick] (6,-6) -- (6,-8) node [midway, left] {$(0)$};
            \draw[very thick] (6,-8) -- (4,-8) node [midway, above] {$(\neg 1)$};
            \draw[very thick] (4,-8) -- (4,-6) node [midway, right] {$(1)$};
            \draw[very thick] (4,-6) -- (6,-8);
            \draw[very thick] (6,-6) -- (4,-8);
            \draw[very thick] (4,-6) -- (6,-6) -- (6,-8) -- (4,-8) -- cycle;
            \draw[very thick] (6,-6) -- (8,-6) node [midway, below] {$(0)$};
            \draw[very thick] (8,-6) -- (8,-8) node [midway, left] {$(0)$};
            \draw[very thick] (8,-8) -- (6,-8) node [midway, above] {$(1)$};
            \draw[very thick] (6,-8) -- (6,-6) node [midway, right] {$(0)$};
            \draw[very thick] (6,-6) -- (8,-8);
            \draw[very thick] (8,-6) -- (6,-8);
            \draw[very thick] (6,-6) -- (8,-6) -- (8,-8) -- (6,-8) -- cycle;
            \draw[very thick] (4,-8) -- (6,-8) node [midway, below] {$(\neg 1)$};
            \draw[very thick] (6,-8) -- (6,-10) node [midway, left] {$(0)$};
            \draw[very thick] (6,-10) -- (4,-10) node [midway, above] {$(1)$};
            \draw[very thick] (4,-10) -- (4,-8) node [midway, right] {$(0)$};
            \draw[very thick] (4,-8) -- (6,-10);
            \draw[very thick] (6,-8) -- (4,-10);
            \draw[very thick] (4,-8) -- (6,-8) -- (6,-10) -- (4,-10) -- cycle;
            \node[anchor=center] at (4, -11.5) {(d)};
		\end{scope}
		\begin{scope}[xshift=12cm,yshift=-13cm]
			\draw[step=2,very thin, dashed] (-1.000000,-11.000000) grid (9.000000,1.000000);
            \draw[very thick] (4,0) -- (6,0) node [midway, below] {$2$};
            \draw[very thick] (6,0) -- (6,-2) node [midway, left] {$2$};
            \draw[very thick] (6,-2) -- (4,-2) node [midway, above] {$0$};
            \draw[very thick] (4,-2) -- (4,0) node [midway, right] {$2$};
            \draw[very thick] (4,0) -- (6,-2);
            \draw[very thick] (6,0) -- (4,-2);
            \draw[very thick] (4,0) -- (6,0) -- (6,-2) -- (4,-2) -- cycle;
            \draw[very thick] (0,-2) -- (2,-2) node [midway, below] {$1$};
            \draw[very thick] (2,-2) -- (2,-4) node [midway, left] {$1$};
            \draw[very thick] (2,-4) -- (0,-4) node [midway, above] {$2$};
            \draw[very thick] (0,-4) -- (0,-2) node [midway, right] {$1$};
            \draw[very thick] (0,-2) -- (2,-4);
            \draw[very thick] (2,-2) -- (0,-4);
            \draw[very thick] (0,-2) -- (2,-2) -- (2,-4) -- (0,-4) -- cycle;
            \draw[very thick] (2,-2) -- (4,-2) node [midway, below] {$2$};
            \draw[very thick] (4,-2) -- (4,-4) node [midway, left] {$1$};
            \draw[very thick] (4,-4) -- (2,-4) node [midway, above] {$0$};
            \draw[very thick] (2,-4) -- (2,-2) node [midway, right] {$1$};
            \draw[very thick] (2,-2) -- (4,-4);
            \draw[very thick] (4,-2) -- (2,-4);
            \draw[very thick] (2,-2) -- (4,-2) -- (4,-4) -- (2,-4) -- cycle;
            \draw[very thick] (4,-2) -- (6,-2) node [midway, below] {$0$};
            \draw[very thick] (6,-2) -- (6,-4) node [midway, left] {$0$};
            \draw[very thick] (6,-4) -- (4,-4) node [midway, above] {$0$};
            \draw[very thick] (4,-4) -- (4,-2) node [midway, right] {$1$};
            \draw[very thick] (4,-2) -- (6,-4);
            \draw[very thick] (6,-2) -- (4,-4);
            \draw[very thick] (4,-2) -- (6,-2) -- (6,-4) -- (4,-4) -- cycle;
            \draw[very thick] (4,-4) -- (6,-4) node [midway, below] {$0$};
            \draw[very thick] (6,-4) -- (6,-6) node [midway, left] {$1$};
            \draw[very thick] (6,-6) -- (4,-6) node [midway, above] {$2$};
            \draw[very thick] (4,-6) -- (4,-4) node [midway, right] {$1$};
            \draw[very thick] (4,-4) -- (6,-6);
            \draw[very thick] (6,-4) -- (4,-6);
            \draw[very thick] (4,-4) -- (6,-4) -- (6,-6) -- (4,-6) -- cycle;
            \draw[very thick] (2,-6) -- (4,-6) node [midway, below] {$(1)$};
            \draw[very thick] (4,-6) -- (4,-8) node [midway, left] {$(1)$};
            \draw[very thick] (4,-8) -- (2,-8) node [midway, above] {$(2)$};
            \draw[very thick] (2,-8) -- (2,-6) node [midway, right] {$(1)$};
            \draw[very thick] (2,-6) -- (4,-8);
            \draw[very thick] (4,-6) -- (2,-8);
            \draw[very thick] (2,-6) -- (4,-6) -- (4,-8) -- (2,-8) -- cycle;
            \draw[very thick] (4,-6) -- (6,-6) node [midway, below] {$2$};
            \draw[very thick] (6,-6) -- (6,-8) node [midway, left] {$(0)$};
            \draw[very thick] (6,-8) -- (4,-8) node [midway, above] {$(\neg 1)$};
            \draw[very thick] (4,-8) -- (4,-6) node [midway, right] {$(1)$};
            \draw[very thick] (4,-6) -- (6,-8);
            \draw[very thick] (6,-6) -- (4,-8);
            \draw[very thick] (4,-6) -- (6,-6) -- (6,-8) -- (4,-8) -- cycle;
            \draw[very thick] (6,-6) -- (8,-6) node [midway, below] {$(0)$};
            \draw[very thick] (8,-6) -- (8,-8) node [midway, left] {$(0)$};
            \draw[very thick] (8,-8) -- (6,-8) node [midway, above] {$(1)$};
            \draw[very thick] (6,-8) -- (6,-6) node [midway, right] {$(0)$};
            \draw[very thick] (6,-6) -- (8,-8);
            \draw[very thick] (8,-6) -- (6,-8);
            \draw[very thick] (6,-6) -- (8,-6) -- (8,-8) -- (6,-8) -- cycle;
            \draw[very thick] (4,-8) -- (6,-8) node [midway, below] {$(\neg 1)$};
            \draw[very thick] (6,-8) -- (6,-10) node [midway, left] {$(0)$};
            \draw[very thick] (6,-10) -- (4,-10) node [midway, above] {$(1)$};
            \draw[very thick] (4,-10) -- (4,-8) node [midway, right] {$(0)$};
            \draw[very thick] (4,-8) -- (6,-10);
            \draw[very thick] (6,-8) -- (4,-10);
            \draw[very thick] (4,-8) -- (6,-8) -- (6,-10) -- (4,-10) -- cycle;
            \node[anchor=center] at (4, -11.5) {(e)};
		\end{scope}
		\begin{scope}[xshift=24cm,yshift=-13cm]
			\draw[step=2,very thin, dashed] (-1.000000,-11.000000) grid (9.000000,1.000000);
            \draw[very thick] (4,0) -- (6,0) node [midway, below] {$2$};
            \draw[very thick] (6,0) -- (6,-2) node [midway, left] {$2$};
            \draw[very thick] (6,-2) -- (4,-2) node [midway, above] {$0$};
            \draw[very thick] (4,-2) -- (4,0) node [midway, right] {$2$};
            \draw[very thick] (4,0) -- (6,-2);
            \draw[very thick] (6,0) -- (4,-2);
            \draw[very thick] (4,0) -- (6,0) -- (6,-2) -- (4,-2) -- cycle;
            \draw[very thick] (0,-2) -- (2,-2) node [midway, below] {$1$};
            \draw[very thick] (2,-2) -- (2,-4) node [midway, left] {$1$};
            \draw[very thick] (2,-4) -- (0,-4) node [midway, above] {$2$};
            \draw[very thick] (0,-4) -- (0,-2) node [midway, right] {$1$};
            \draw[very thick] (0,-2) -- (2,-4);
            \draw[very thick] (2,-2) -- (0,-4);
            \draw[very thick] (0,-2) -- (2,-2) -- (2,-4) -- (0,-4) -- cycle;
            \draw[very thick] (2,-2) -- (4,-2) node [midway, below] {$2$};
            \draw[very thick] (4,-2) -- (4,-4) node [midway, left] {$1$};
            \draw[very thick] (4,-4) -- (2,-4) node [midway, above] {$0$};
            \draw[very thick] (2,-4) -- (2,-2) node [midway, right] {$1$};
            \draw[very thick] (2,-2) -- (4,-4);
            \draw[very thick] (4,-2) -- (2,-4);
            \draw[very thick] (2,-2) -- (4,-2) -- (4,-4) -- (2,-4) -- cycle;
            \draw[very thick] (4,-2) -- (6,-2) node [midway, below] {$0$};
            \draw[very thick] (6,-2) -- (6,-4) node [midway, left] {$0$};
            \draw[very thick] (6,-4) -- (4,-4) node [midway, above] {$0$};
            \draw[very thick] (4,-4) -- (4,-2) node [midway, right] {$1$};
            \draw[very thick] (4,-2) -- (6,-4);
            \draw[very thick] (6,-2) -- (4,-4);
            \draw[very thick] (4,-2) -- (6,-2) -- (6,-4) -- (4,-4) -- cycle;
            \draw[very thick] (4,-4) -- (6,-4) node [midway, below] {$0$};
            \draw[very thick] (6,-4) -- (6,-6) node [midway, left] {$1$};
            \draw[very thick] (6,-6) -- (4,-6) node [midway, above] {$2$};
            \draw[very thick] (4,-6) -- (4,-4) node [midway, right] {$1$};
            \draw[very thick] (4,-4) -- (6,-6);
            \draw[very thick] (6,-4) -- (4,-6);
            \draw[very thick] (4,-4) -- (6,-4) -- (6,-6) -- (4,-6) -- cycle;
            \draw[very thick] (2,-6) -- (4,-6) node [midway, below] {$1$};
            \draw[very thick] (4,-6) -- (4,-8) node [midway, left] {$1$};
            \draw[very thick] (4,-8) -- (2,-8) node [midway, above] {$2$};
            \draw[very thick] (2,-8) -- (2,-6) node [midway, right] {$1$};
            \draw[very thick] (2,-6) -- (4,-8);
            \draw[very thick] (4,-6) -- (2,-8);
            \draw[very thick] (2,-6) -- (4,-6) -- (4,-8) -- (2,-8) -- cycle;
            \draw[very thick] (4,-6) -- (6,-6) node [midway, below] {$2$};
            \draw[very thick] (6,-6) -- (6,-8) node [midway, left] {$0$};
            \draw[very thick] (6,-8) -- (4,-8) node [midway, above] {$2$};
            \draw[very thick] (4,-8) -- (4,-6) node [midway, right] {$1$};
            \draw[very thick] (4,-6) -- (6,-8);
            \draw[very thick] (6,-6) -- (4,-8);
            \draw[very thick] (4,-6) -- (6,-6) -- (6,-8) -- (4,-8) -- cycle;
            \draw[very thick] (6,-6) -- (8,-6) node [midway, below] {$0$};
            \draw[very thick] (8,-6) -- (8,-8) node [midway, left] {$0$};
            \draw[very thick] (8,-8) -- (6,-8) node [midway, above] {$1$};
            \draw[very thick] (6,-8) -- (6,-6) node [midway, right] {$0$};
            \draw[very thick] (6,-6) -- (8,-8);
            \draw[very thick] (8,-6) -- (6,-8);
            \draw[very thick] (6,-6) -- (8,-6) -- (8,-8) -- (6,-8) -- cycle;
            \draw[very thick] (4,-8) -- (6,-8) node [midway, below] {$2$};
            \draw[very thick] (6,-8) -- (6,-10) node [midway, left] {$0$};
            \draw[very thick] (6,-10) -- (4,-10) node [midway, above] {$1$};
            \draw[very thick] (4,-10) -- (4,-8) node [midway, right] {$0$};
            \draw[very thick] (4,-8) -- (6,-10);
            \draw[very thick] (6,-8) -- (4,-10);
            \draw[very thick] (4,-8) -- (6,-8) -- (6,-10) -- (4,-10) -- cycle;
            \node[anchor=center] at (4, -11.5) {(f)};
		\end{scope}		
	\end{tikzpicture}
}
    \caption{Illustration of the tree solver. First, (a) we pick input legs on the leaves of the tree and an output leg on the root of the tree (here it is the top cell). Then, we propagate the conditions to the nodes in the tree (b) and to the root (c). If the conditions on the root are compatible, we propagate the solution back from the root (d) to the nodes (e) and the leaves (f).}
    \label{fig:tree_alg}
\end{figure}

\begin{ex}\label{ex:result}
Figure~\ref{fig:ex-smiley} shows an example of hand-filled regions (the background, the eyes and the mouth of the smiley),
completed with an automatically tiled region (the face of the smiley).
\begin{figure}
    \centering
    \includegraphics[width=0.2\linewidth]{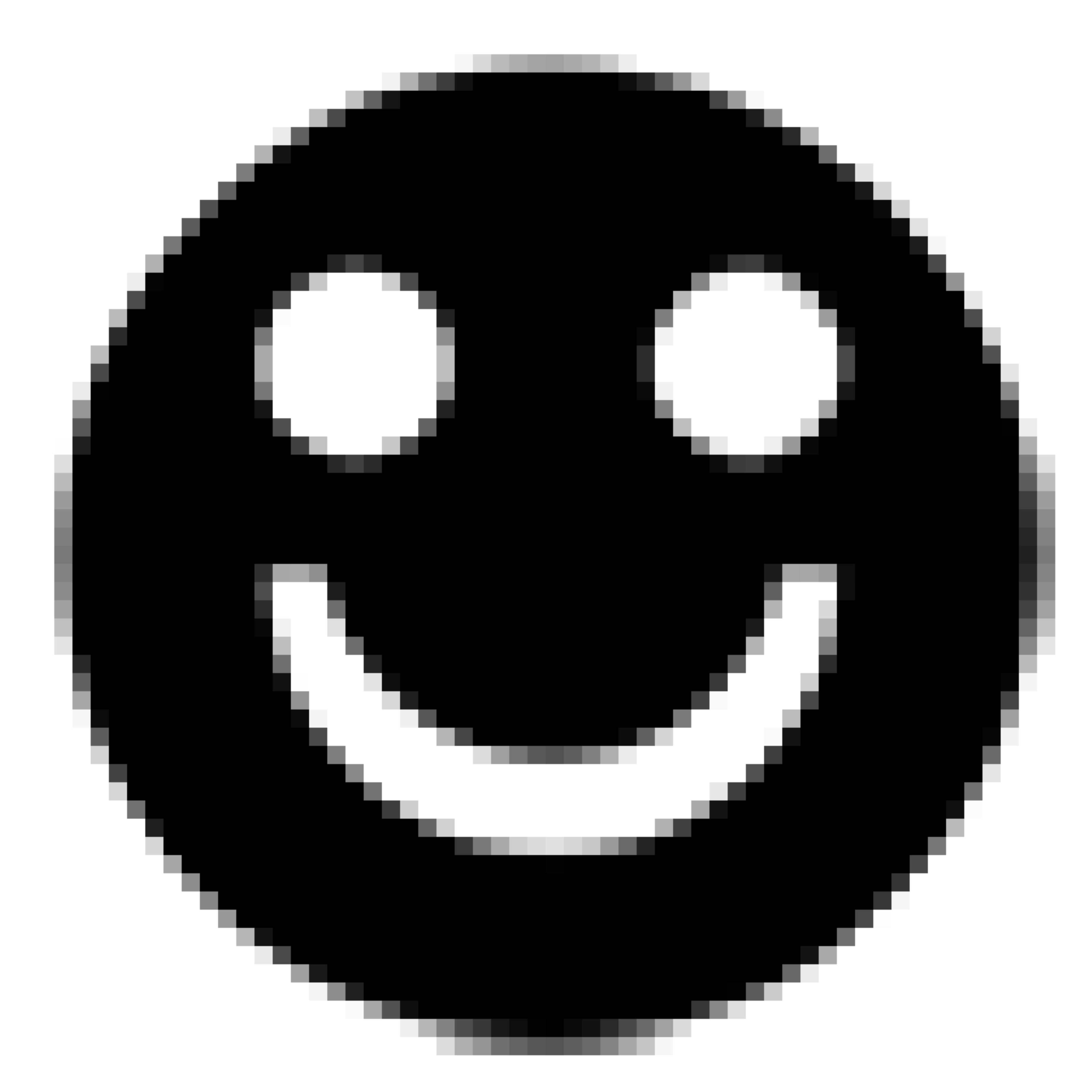}\hspace{1cm}
    \includegraphics[width=0.35\linewidth]{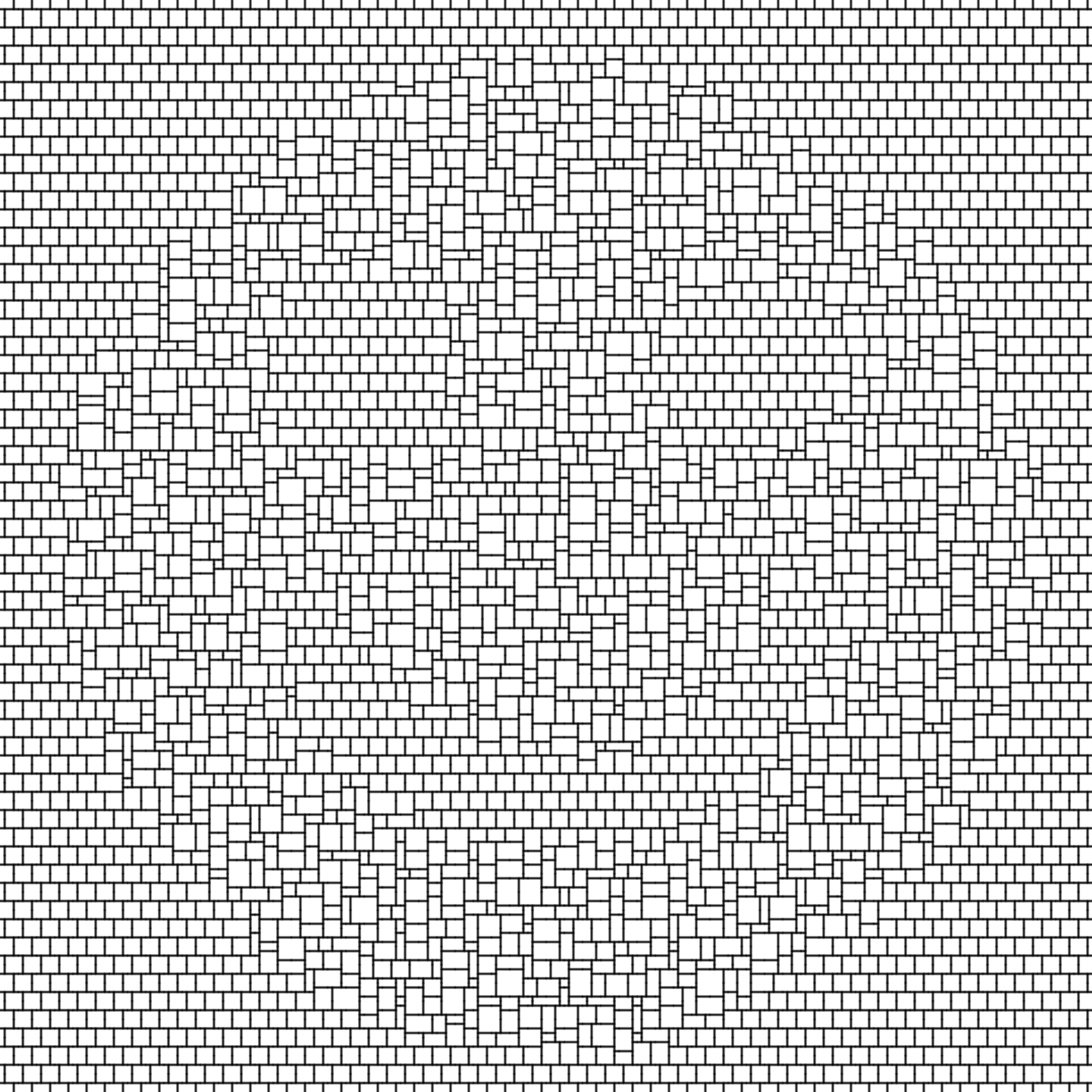}
    \caption{Example of a tiling of a given pattern (top image). The background, the eyes and the mouth are tiled by hand with a regular pattern and the rest (face of the smiley) is tiled automatically.}
    \label{fig:ex-smiley}
\end{figure}
\end{ex}

%%%%
\section{Conclusion and Future work}\label{Conclusion}
In this paper, we have proposed a graph theoretical approach 
to the Wang tiling problem.
We have introduced a property of Wang tile sets 
which is often possessed by those appearing in computer graphics applications.
For the class of Wang prototile sets with this property,
we have devised a general strategy to reduce a tiling problem 
to its sub-problem.
Then, we used the general result to give
a linear algorithm for the tiling of planar regions with 
the brick Wang tiles.
In future work, 
we would like to consider tiling problems 
on hexagonal/triangular and irregular polygonal surface meshes, and also volumetric meshes,
which would be useful for texturing and point sampling on meshes.
Also, we would like to discuss enumeration of all valid tilings
for a given problem.

\section*{Acknowledgements}
We are grateful to the anonymous referees for their fruitful comments.
The second named author was partially supported by  JST PRESTO Grant Number JPMJPR16E3, Japan.

%\section*{References}
%\nocite{*}
\bibliographystyle{plain}

\end{document}